\def\Bbb{\mathbb}
\def\sphere{S^2}

\def\dist{{\rm{dist}}}

\def\dist{{\rm{dist}}}
\def\diam{{\rm{diam}}}

\def\Bbb{\mathbb}
\def\reals{\Bbb R}
\def\uhp{\Bbb H}
\def\uhs{{\Bbb H}^3_+}

\def\disk{\Bbb D}
\def\circle{\Bbb T}

\def\complexes{\Bbb C}

\def\cal{\mathcal}

\def\z{{\bf z}}

\documentclass[12pt]{amsart}  
\usepackage{graphicx,amssymb,amsthm,verbatim,amsfonts,amscd}

\theoremstyle{plain}                    
\newtheorem{thm}{Theorem}[section]

\newtheorem{lemma}[thm]{Lemma}

\newcounter{ques}
   
\numberwithin{equation}{section}

\begin{document}
\baselineskip=18pt


%

\title [Angle bounds for quadrilateral meshes]
  {Optimal angle bounds for quadrilateral meshes}

\author {Christopher J. Bishop}
\address{C.J. Bishop\\
         Mathematics Department\\
         SUNY at Stony Brook \\
         Stony Brook, NY 11794-3651}
\email {bishop@math.sunysb.edu}
\thanks{The  author is partially supported by NSF Grant DMS 04-05578.}

\date{}
\maketitle


\begin{abstract}
We show that any simple planar $n$-gon can be meshed in linear 
time  by $O(n)$ quadrilaterals with all new angles 
bounded between $60$ and $120$ degrees.
\end{abstract}

\clearpage


\setcounter{page}{1}
\renewcommand{\thepage}{\arabic{page}}

\section{Introduction} \label{intro}

We answer a question of Bern and Eppstein by proving:

\begin{thm} \label{main}
Any
simply connected planar domain $\Omega$ whose boundary is a simple $n$-gon 
  has a quadrilateral mesh
with $O(n)$ pieces so that  all  angles are  between 
$60^\circ $ and $120^\circ$, except  that original angles of the
polygon with angle $< 60^\circ$ remain. The mesh can be constructed in 
time $O(n)$.
\end{thm}

The theorem is sharp in the sense that no shorter interval of 
angles suffices for all polygons: using Euler's formula, 
Bern and Eppstein proved 
(Theorem 5 of \cite{Bern-Eppstein-QuadMesh}) that any 
quadrilateral mesh of a  polygon
with all angles $\geq 120^\circ$ 
must contain an angle $\geq 120^\circ$. On the other hand, any 
boundary angle $\theta > 120^\circ$ must be subdivided by the mesh 
in Theorem \ref{main}
and hence there must be a new  angle $\leq \theta/2$ in the mesh.
Thus taking  polygons with an angle $ \theta \searrow 120^\circ$ 
 shows $60^\circ $ is the optimal lower bound.

It is perhaps best to think of Theorem \ref{main} as an 
existence result. Although we give a linear time algorithm 
for finding the mesh, the constant is large and the construction 
depends on other linear algorithms, such Chazelle's linear time
triangulation of  polygons, that  have  not been  implemented (as far as
I know).

The three main tools  in the proof of Theorem \ref{main}  are 
 conformal maps, thick/thin decompositions of
 polygons and hyperbolic tesselations.
We will decompose $\Omega$ into $O(n)$ ``thick'' and ``thin'' parts.
The thin parts have simple shapes and we can easily construct 
an explicit mesh in each of them. The thick parts are more complicated, 
but we can use a conformal map to  transfer a mesh from 
the unit disk, $\disk$, to the thick parts of $\Omega$
 with small distortion.
The mesh on $\disk $ is produced using a finite piece
 of an infinite tesselation 
of $\disk$ by hyperbolic  pentagons.

I would like to thank Marshall Bern for asking me the question that lead
to Theorem \ref{main} and  pointing out his paper \cite{Bern-Eppstein-QuadMesh} 
with David Eppstein.
Also thanks to Joe Mitchell for many helpful conversations on 
computational geometry.  This paper is part of a series
(\cite{Bishop-Bowen},
\cite{Bishop-BrenConj}, 
\cite{Bishop-ExpSullivan}, 
\cite{Bishop-time})
that exploits the close connection between 
 the medial axis of a planar domain, the geometry of its hyperbolic
convex hull in $\uhs$  and the conformal map of the domain to the disk.
This was originally motivated by a result of Dennis Sullivan  \cite{Sullivan81} 
  about boundaries of hyperbolic 3-manifolds and its generalization by 
 David Epstein (only one ``p'' this time) and Al Marden \cite{EM87}.
Many thanks to those authors for the inspiration and insights they have provided.
Also many thanks to the referees for a careful reading of the original 
manuscript. Their thoughtful comments and suggestions greatly improved 
the paper. 
One of them pointed out reference \cite{Gerver} where 
the Riemann mapping theorem 
is used to prove that any polygon with all angles $\geq \pi/5$ can be 
dissected into triangles with all angles $\leq 2 \pi /5$.


\section{M{\"o}bius transformations and hyperbolic geometry} \label{mobius}

A linear fractional (or M{\"o}bius)  transformation is 
a map of the form $z \to (a z + b )/ ( c z + d)$. This is a 1-1, onto, 
holomorphic map of the Riemann sphere $\sphere = 
\complexes \cup \{ \infty \}$ to 
itself.  Such maps form a group under composition and
are well known to map circles to circles (if we 
count straight lines as circles that pass through $\infty$).
M{\"o}bius transforms are conformal, so they preserve angles.
Given two sets of  distinct points $\{ z_1, z_2, z_3\}$ and 
$\{ w_1, w_2, w_3 \}$ there is a unique M{\"o}bius transformation that
sends $w_k \to z_k$ for $k=1,2,3$.
A M{\"o}bius transformation maps the unit disk, $\disk$, to itself 
iff it is of the form $g(z) =\lambda (z-a)/(1-\bar az)$ for some $a \in \disk, 
|\lambda|=1$.

The hyperbolic metric on the unit disk is given 
by 
$$ \rho(v,w) = \inf \int_\gamma \frac {2 |dz|}{1-|z|^2},$$
where the infimum is over all rectifiable arcs connecting 
$v$ and $w$ in $\disk$. This is a metric of constant negative 
curvature. In some sources, the ``2'' is omitted; 
we have chosen this version to be consistent with 
the  trigonometric formulas found in \cite{Beardon}.
Geodesics  for this metric are circular arcs that are perpendicular 
to the boundary (including diameters).
Hyperbolic area is given by $4 dxdy /(1-|z|^2)^2$.
The area of a triangle with geodesic edges is $\pi-\alpha -
\beta-\gamma$, where $\alpha, \beta, \gamma$ are the interior 
angles. Thus the area of any hyperbolic triangle is $\leq \pi$.

The hyperbolic metric 
is well known to be invariant
under M{\"o}bius transformations of the disk, so it is enough 
to compute it when one point has been normalized to be $0$ and 
the other rotated to the positive axis.  If $0< x <1$ and $\rho = 
\rho(0,x)$, then 
$$ \rho= \log \frac {1+x}{1-x}, \qquad x = \frac {e^\rho -1}{e^\rho +1}.$$
It is also convenient to consider the isometric model of the 
upper half-space, $\uhp$. In this case the hyperbolic  metric is given 
by 
$$ \rho(v,w) = \inf \int_\gamma \frac { |dz|}{y},$$
 where $ z = x + iy$, but geodesics are still circular arcs 
perpendicular to the boundary.  

If $E \subset   \circle = \partial \disk$ is closed then 
$\circle \setminus E = \cup I_j$ is a union of open intervals. 
The hyperbolic convex hull of $E$, denoted ${\rm{CH}(E)}$,  
is the region in $\disk$ bounded by $E$ and the collection of 
circular arcs $\{ \gamma_j\}$, where $\gamma_j$ is the hyperbolic 
geodesic with the same endpoints as $I_j$. See Figure \ref{W}.
\begin{figure}[htbp]
\centerline{
 \includegraphics[height=1.5in]{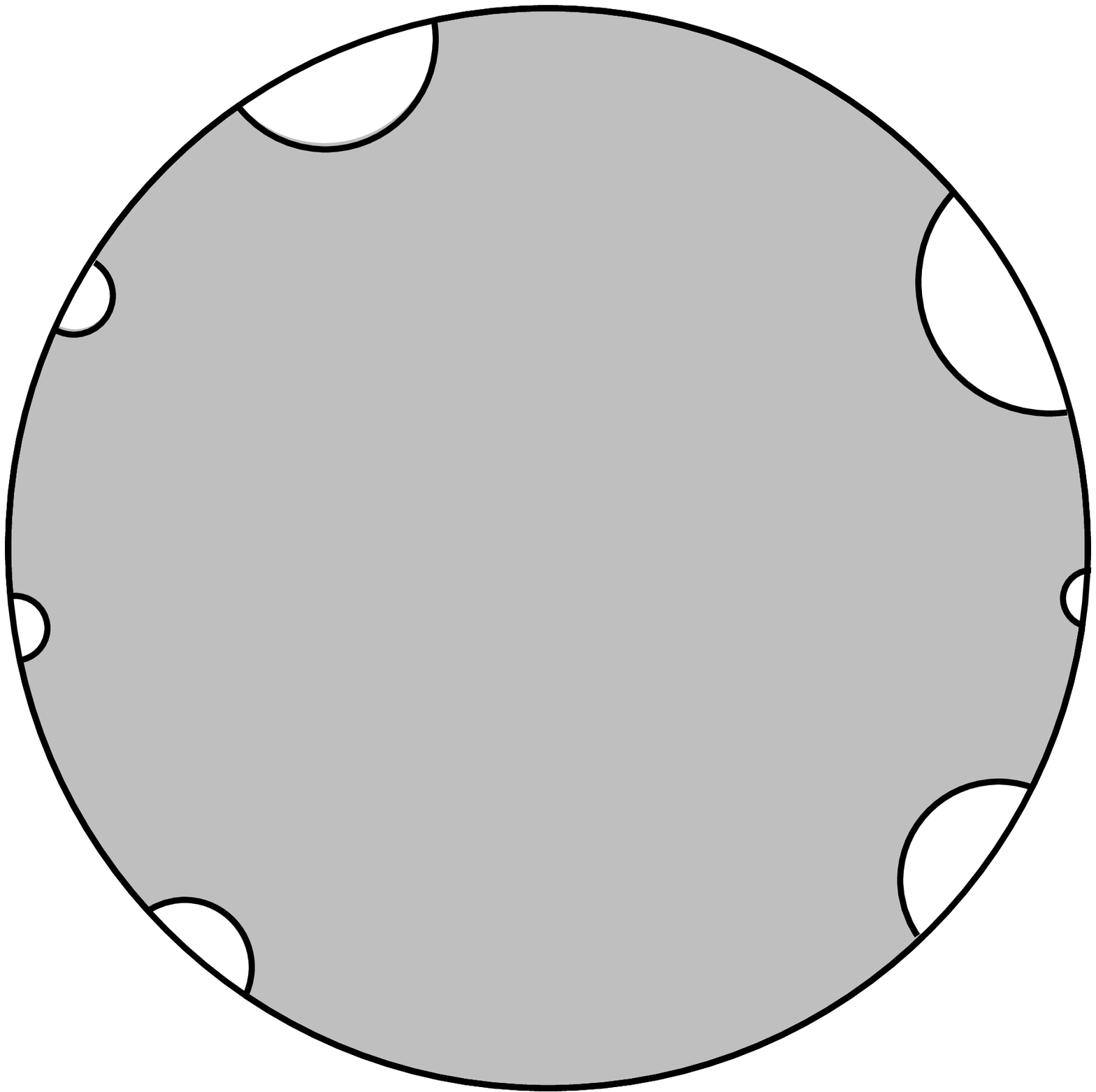}
 $\hphantom{xxxxx}$
 \includegraphics[height=1.5in]{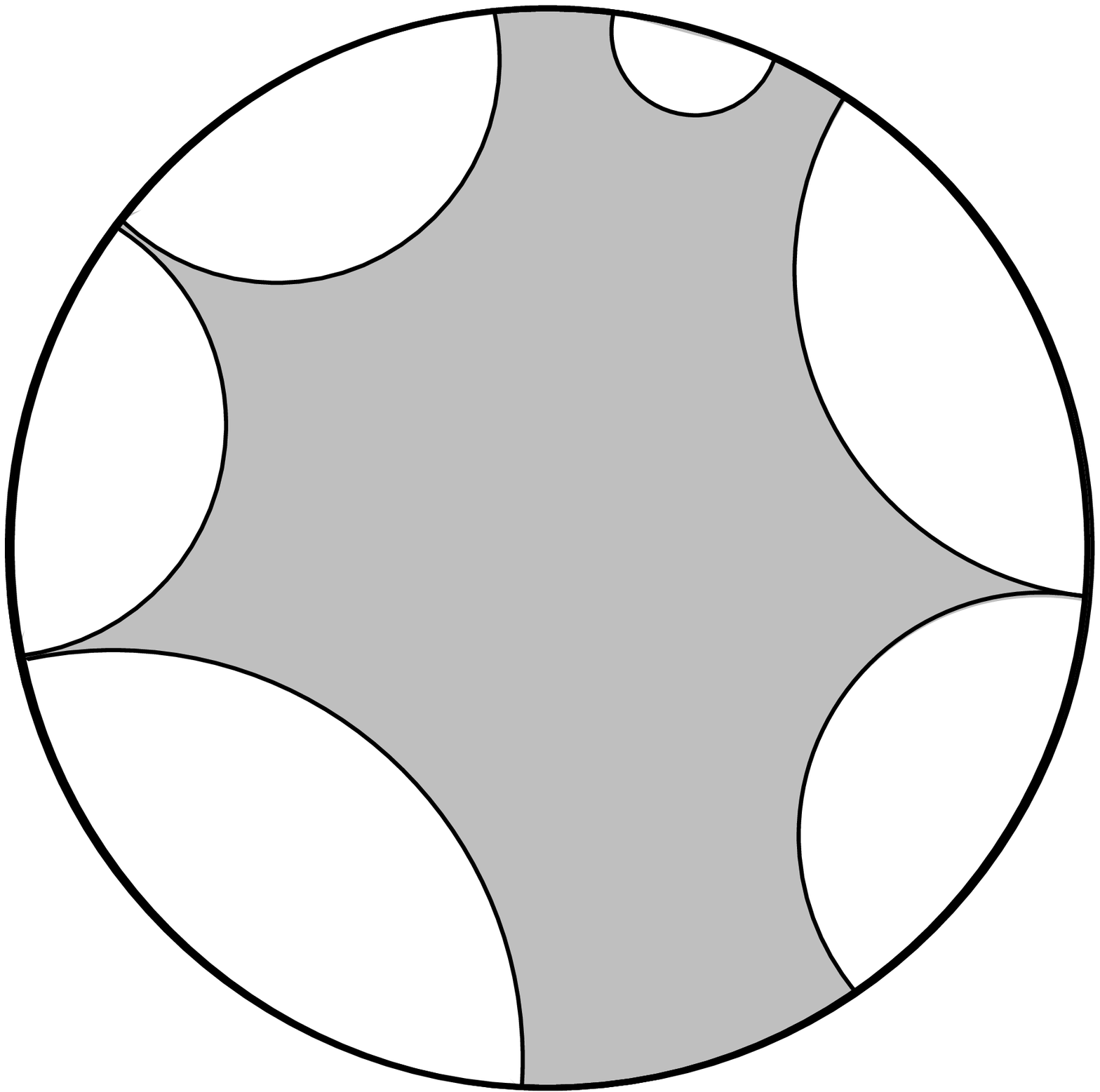}
 $\hphantom{xxxxx}$
 \includegraphics[height=1.5in]{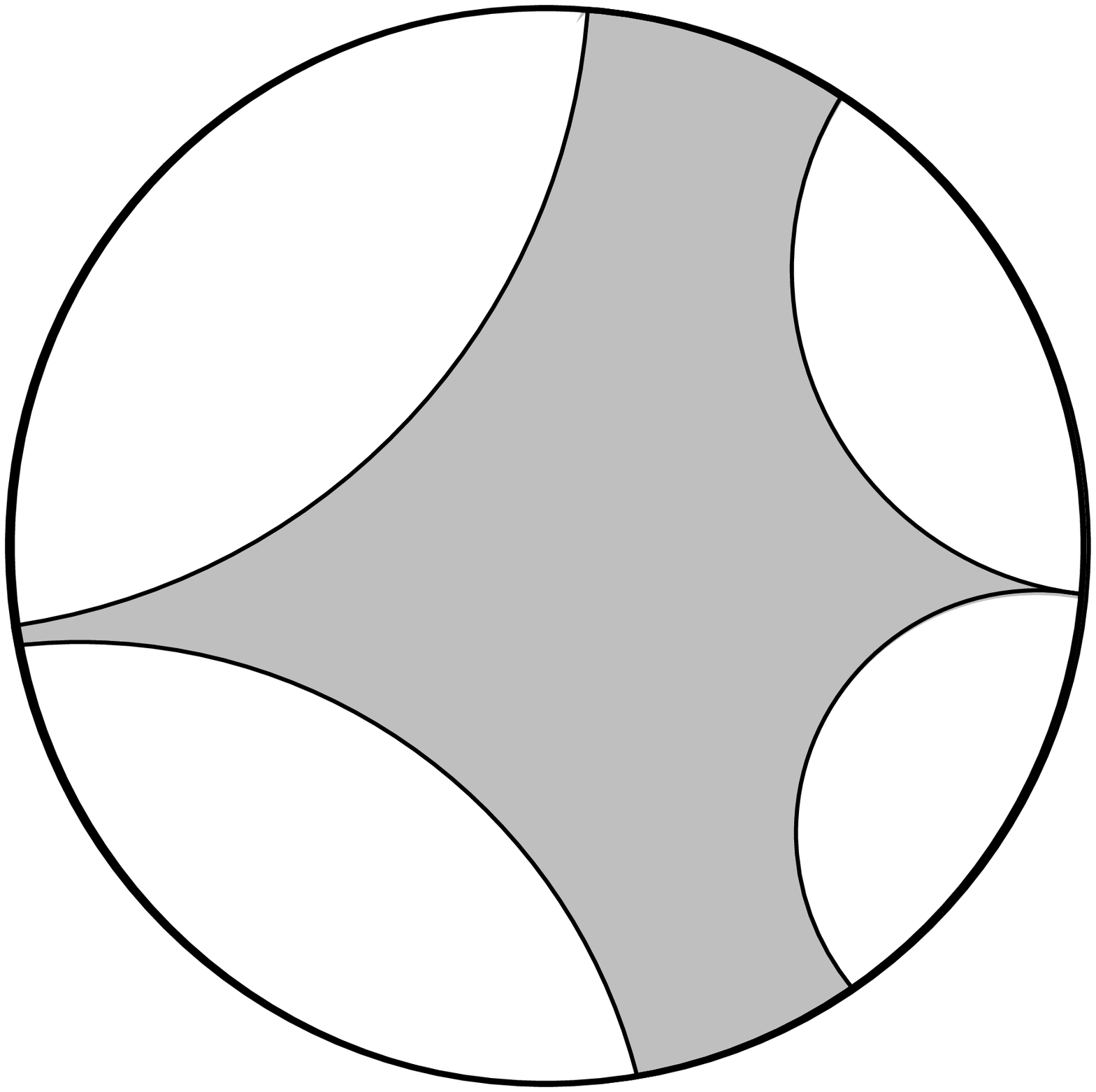}
 }
\caption{ \label{W} 
Examples of  hyperbolic convex hulls. The one on the left is 
uniformly perfect, the center is thick  with a large $\eta$, but not 
uniformly prefect, and the right is only thick with a small $\eta$
(there are two geodesics that almost touch, but do not share an endpoint).
}
\end{figure}

A closed  set $E \subset \circle$ is called $\eta$-thick if 
any two components of $\partial {\rm{CH}}(E) \cap \disk$ that don't share
an endpoint are at least hyperbolic distance  $\eta$ apart.
If $E$ is $\eta$-thick, then  any point in the hull is contained in a 
hyperbolic ball of radius $\eta$ that is also contained in the 
convex hull. The thickness condition can be written in other 
ways. For example, $E$ is $\eta$-thick iff non-adjacent 
complementary intervals have extremal distance at least $\delta>0$
(with $\delta^{-1} \simeq \frac 2\pi \log \frac 1 \eta$ for small 
$\delta, \eta$) \cite{Bishop-time}.
 A closed set $E$ is called uniformly perfect  
if any two  components of  $\partial {\rm{CH}}(E) \cap \disk$  are at least 
hyperbolic distance $\eta $ part. This stronger condition arises 
many places in function theory, but will not be used in this paper.


\newpage 

\section{A subdivision of the hyperbolic disk} \label{subdivide disk}

To prove Theorem \ref{main}  we will divide  the interior of 
$\Omega$ into pieces  called ``thick'' and ``thin'' (see 
\cite{Bishop-time} and Section \ref{thick and thin}).
The thin pieces will be meshed 
explicitly, but the  mesh on the thick pieces will
 be transferred from 
a quadrilateral mesh of a domain  in the unit disk via a conformal map.
Most of our time will be spent  constructing the
mesh on the disk. 
In this section we describe the subdomain and how to subdivide it into 
circular arc
triangles, quadrilaterals and pentagons. In  the following 
sections we 
show how to construct quadrilateral meshes for each  subregions that 
are consistent along shared boundaries.

A compact hyperbolic polygon is a bounded region in hyperbolic space bounded by 
a finite number geodesic segments.  The polygon is ``right'' if every 
interior angle is 
$90^\circ$.  There are no compact hyperbolic right triangles or 
quadrilaterals, but there are hyperbolic right $n$-gons 
for every $n \geq 5$  and any such can be extended to a tesselation 
${ \cal T}_n$ of hyperbolic space by repeated reflections. 
 See Figure  \ref{n-gons} for the case of pentagons (the only 
case we use in this paper).  

\begin{figure}[htbp]
\centerline{
 \includegraphics[height=1.5in]{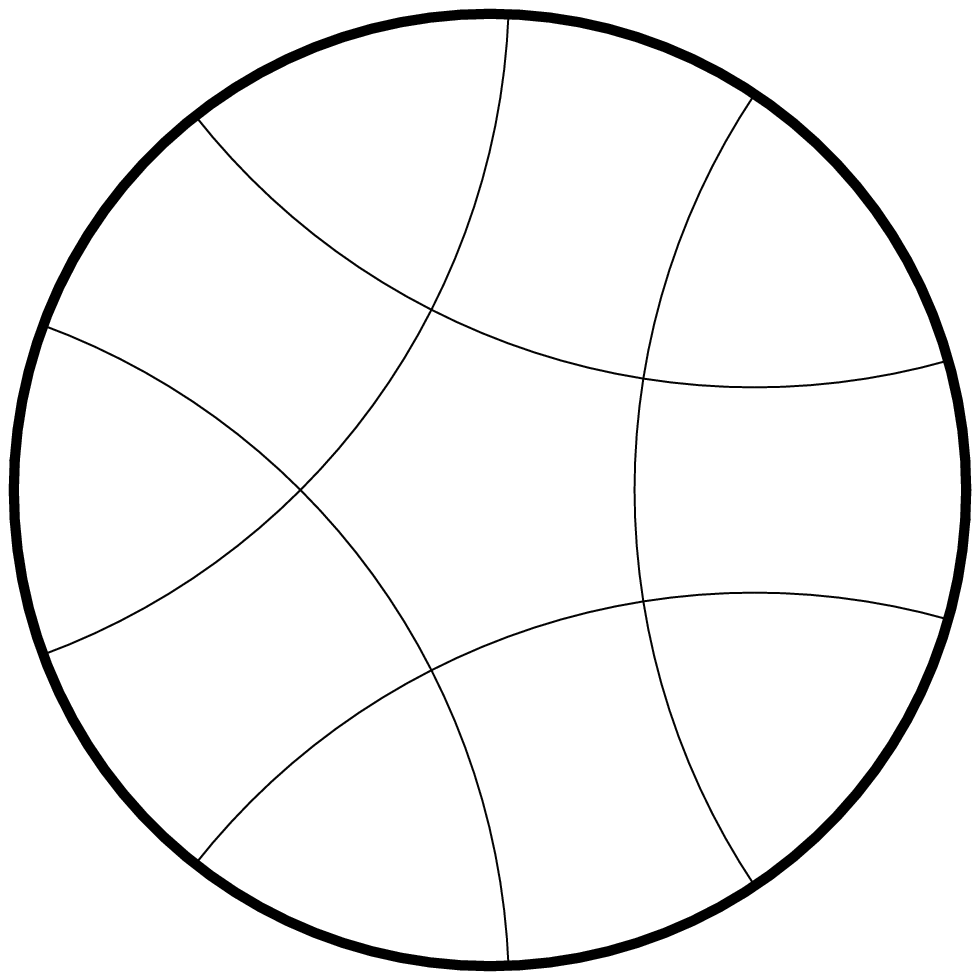}
$\hphantom{xxxx}$
 \includegraphics[height=1.5in]{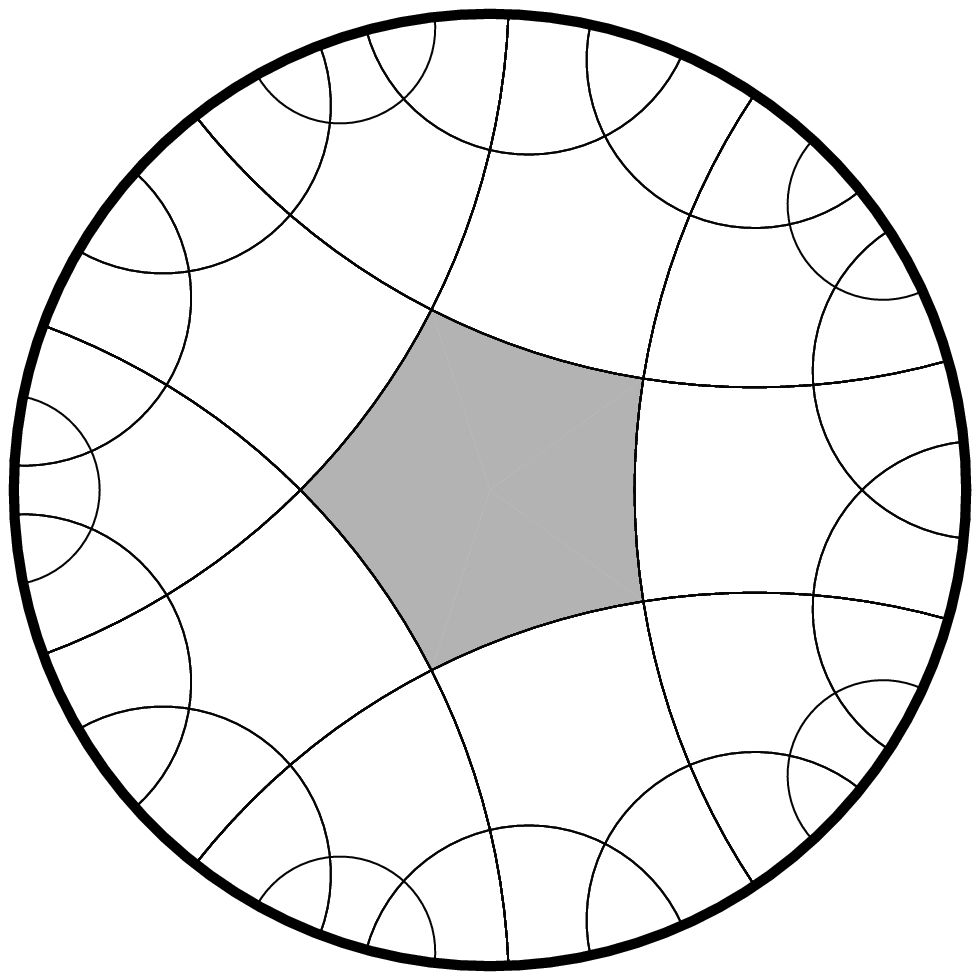}
 }
\caption{ \label{n-gons} 
A hyperbolic right  pentagon (left) and the its neighbors in the 
tesselation ${\cal T}_5$.
}
\end{figure}

Let $L= \cosh^{-1}(1 + 2 \cos(\frac {2 \pi }{5})) \approx 1.06128$
  denote the side length of a hyperbolic right 
pentagon. We don't need the specific value, but it can 
be computed using $L=c$,  $\gamma= 2\pi/5$, $\alpha= \beta = \pi/4$
in the   second hyperbolic law of cosines
(see \cite{Beardon}):
$$ \cosh c = \frac {\cos \alpha \cos \beta + \cos \gamma}
                    {\sin \alpha \sin \beta}.$$

  In the  tesselation ${\cal T}_5$, each edge 
of a pentagon lies on some hyperbolic geodesic.  Each of these 
geodesics divides $\circle$ into two arcs and we let 
${\cal I}_5$ denote the collection of all such arcs.

\begin{lemma} \label{cover 1} 
There is a $c < \infty$ so that given any arc $J \subset \circle$
there are $I_1, I_2 \in {\cal I}_5$ with $I_1 \subset J \subset 
I_2$ and $|I_2|/|I_1| \leq c$ ($|\cdot|$ denotes arclength).
\end{lemma}

\begin{proof}
Let $\gamma$ be the hyperbolic geodesic with the same endpoints 
as $J$. The top point of $\gamma$ (i.e., the point closest
to $0$) is contained in some pentagon of the tesselation.  By taking 
$c$ larger, we can assume $J$ is as short as we wish, so we 
may assume this is not the central pentagon.  Let 
$a$ be the hyperbolic center of this pentagon and let 
$g(z) = \lambda (z-a)/(1- \bar a z)$ where $|\lambda|=1$
is chosen $g$ maps the pentagon to the central pentagon.
 This is a M{\"o}bius 
transformation that sends $a$ to $0$, maps the diameter  $D$
through $a$ into  $\lambda D$ and maps $\gamma$ to a geodesic $\gamma'$ that 
intersects the central pentagon of the tesselation. Moreover, 
since $g$ preserves angles,  the angle between $\gamma'$ and 
$ D'= \lambda D$ is the same as between $\gamma$ and $D$, and this is 
bounded away from $0$, since the intersection point is within 
distance $L$ of the top point of $\gamma$.

Thus $\gamma'$ also makes a large angle with $ D'$ and so is 
some positive distance $r$ from the point $b= -\lambda a=g(0)$.
The inverse of $g$ 
is $ f(z) = \bar \lambda (z-b)/(1-\bar b z)$ and the 
derivative of this is  $(1-|b|^2)/(1-\bar b z)^2$. From this
we see that for $|z| =1$, 
$$ \frac {1-|b|}{|z-b|^2} \leq 
      | f'(z) | \leq \frac {2(1-|b|)}{|z-b|^2}|,$$
so that $ | f'(z) |  \simeq 2(1-|b|) $ with a constant
that depends only on $|z-b|$.  Thus sets outside a ball 
around $b$ will be compressed similar amounts by $f$.

Choose geodesics $\gamma_1, \gamma_2$ from the tesselation edges
on either side of $\gamma'$ so that $\gamma_1$ separates $b$
from $\gamma'$ and has a uniformly bounded distance $r$ from 
$b$ (we can easily do this if $1-|b| = 1-|z| \simeq |J|$ is small 
enough). Apply $f$ to $\gamma_1, \gamma_2$ and we get two
geodesics of comparable Euclidean size  whose  base intervals 
are the desired $I_1, I_2$. 
\end{proof}

A Carleson triangle in $\disk$ is a region bounded by two geodesic 
rays that have a common endpoint where they meet with 
interior angle $90^\circ$. Any two such are M{\"o}bius 
equivalent. A Carleson quadrilateral is bounded by one 
finite length hyperbolic segment and two geodesic rays, 
again with both interior angles equal $90^\circ$. See
Figure  \ref{HCQ}. It is determined up to isometry by the 
hyperbolic length of its finite length edge. In this paper 
all of our Carleson quadrilaterals with have length $L$, where
$L$ is the side length of a right pentagon, as above.

\begin{figure}[htbp]
\centerline{ 
\includegraphics[height=1.0in]{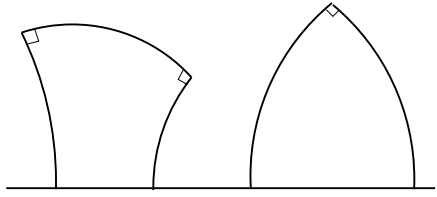}
}
\caption{ \label{HCQ} 
A Carleson quadrilateral and triangle. 
}
\end{figure}

We will prove the following:

\begin{lemma} \label{build intervals}
There is a $c < \infty$ so that the following holds. 
Suppose we are given $A>1$  and a finite collection 
 intervals $\{ I_j\}_1^N$ on 
the unit circle so that the expanded intervals $\{ A I_j\}$ are
disjoint (these are the concentric intervals that are $A$ times longer)
and each has length $< \pi$. Let $E = \cup_j I_j$. 
  We can find intervals $\{ J_j\}$ so that
\begin{enumerate}
 \item  $  \sqrt{A} I_j \subset 
J_j \subset     c \sqrt{A}  I_j$, $j=1, \dots, N$.  
   \item  Let $F = \cup_j J_j$ and let $W \subset \disk$ be the hyperbolic 
   convex hull of $\circle \setminus F$. Then $W$ 
      has a mesh  $\{ W_k\}$ consisting of 
     right hyperbolic pentagons,  Carleson quadrilaterals and  
     Carleson triangles. A pentagon shares an edge only with other pentagons 
     or the top of a  quadrilateral, a quadrilateral shares a top edge
    only with pentagons and  side edges
       with  triangles and other quadrilaterals,
       and a triangle shares edges only with quadrilaterals. 
   \item Each component  of $\partial  W \cap \disk$ 
       is an infinite geodesic that is the  union of  side edges from 
     two Carleson quadrilaterals and edges from three pentagons.
   \item Every pentagon used in the mesh  is a uniformly bounded 
    hyperbolic distance from the hyperbolic convex hull of $E$.
   \item Every region  $W_k$ in the mesh has diameter bounded by
        $O(\dist(W_k, E))$ (Euclidean distances). 
\end{enumerate}
\end{lemma}


\begin{proof} 
For each interval  $ I_j$ given in  the lemma, choose
$ J_j \in {\cal I}_5 $ to be the minimal interval 
containing $\sqrt{A} I_j$. Then (1) clearly holds by 
Lemma \ref{cover 1}. 

Let $\gamma_j$ be the geodesic with the same endpoints 
as $J_j$ and let $P_0$ be a  pentagon in ${\cal T}_5$ 
that is above $\gamma_j$ (i.e, 
whose interior lies in the component of $\disk \setminus 
\gamma_j$ containing $0$) and  whose boundary contains the 
``top'' of $\gamma_j$ (the point closest to $0$).
Let $P_1,P_2$ be the elements of ${\cal T}_5$ that are 
adjacent to $P_0$ and also above  $\gamma_j$. Then the 
part of $\gamma_j$ covered by the boundaries of these 
three pentagons contains an  interval of hyperbolic length 
$2L$ centered at the top point.   
Let $\gamma_j^1 $ be the geodesic containing  the side of 
$P_1$ that has one endpoint on $\gamma_j$ and is 
not on $\partial P_0$. Let $J_j^1 \in {\cal I}_5$
 be the base interval of $\gamma_j^1$.  Let 
$J_j^2 \in {\cal I}_5$ be the corresponding interval 
for $P_2$ and let $J_j' = J_j^1 \cup J_j \cup J_j^2$.
See Figure \ref{Jinterval}.

\begin{figure}[htbp]
\centerline{ 
\includegraphics[height=1.5in]{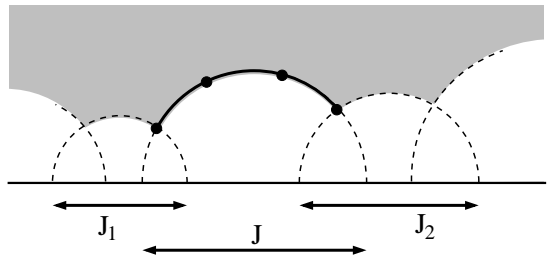}
}
\caption{ \label{Jinterval} 
On the left are $J, J_1, J_2$. The shaded region is a 
union of the pentagons; the white is a union of 
quadrilaterals and triangles.
}
\end{figure}

Let $G= \cup_j (J_j^1 \cup J_j \cup J_j^2)$
and let $\{\cal K\}$ be the collection of 
intervals in ${\cal I}_5$ that are  compactly
contained in $\circle \setminus F$, contain a 
point of $\circle \setminus G$  and 
are maximal in the sense of containment with respect to these
properties.
These clearly cover all of $\circle \setminus G$.
Now add the intervals $J_j, J_j^1, J_j^2$ to get 
a cover of the whole circle. 
Any open finite cover of an interval   has
 a subcover with overlaps  of at most $2$ (if a point 
is in three intervals we can keep the ones with leftmost left 
endpoint and rightmost right endpoint and throw away the third; 
repeat until every point is in at most two intervals). 
For such a subcover, we mesh $W$ with pentagons above the 
corresponding geodesics and by Carleson quadrilaterals and 
triangles below. See Figure \ref{W4}. 
 Conditions (2) and (3) are  clear from construction.   
\begin{figure}[htbp]
\centerline{ 
\includegraphics[height=1.75in]{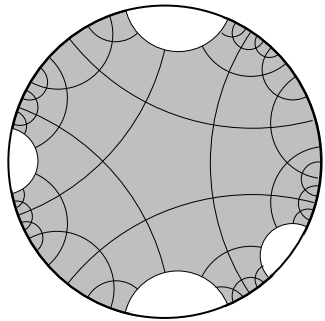}
}
\caption{ \label{W4} 
An example of  meshing a convex hull  $W$ with pentagons, quadrilaterals and 
triangles.  This example  is not to scale,  
since  the white regions 
should be much smaller than their distances apart.
}
\end{figure}

If $x \in \circle \setminus F$ and $d = \dist (x, F)$ then 
apply Lemma \ref{cover 1} to an interval of length 
$ \frac 12 d/c$ centered at $x$. We obtain an element of ${\cal I}_5$ 
containing $x$, missing $F$ and of length $\geq  \frac 12 d/c $. 
Thus the maximal interval of ${\cal K}$ containing $x$ has at least
this length. This implies (4).

Every right pentagon  $P$ has Euclidean 
diameter bounded by $O(\dist(P, \circle))= O(\dist(P,E))$. 
Every Carleson quadrilateral $R$ has a top edge along a geodesic  $\gamma$
with endpoints $\{ a,b \}$ and $\diam(R) \simeq \dist(R, \{ a,b\})$.
Since $\gamma$ misses  the hyperbolic 
convex hull of $E$,  the latter is $\leq \dist(Q, R)$.
Every Carleson triangle is adjacent to two Carleson quadrilaterals 
of comparable Euclidean size that separate it from $E$, so the 
estimate also holds for these triangles. Thus (5) holds.
\end{proof} 

\begin{lemma} \label{thick collections} 
If the collection $\{ I_j\}_1^n$  satisfies the 
conditions of Lemma \ref{build intervals} and if, in addition, the
set $E=\cup_j I_j$ is a $\delta$-thick set, 
then the mesh constructed in Lemma \ref{build intervals} has $O(n)$
elements, with a constant that depends only on $\delta$. 
\end{lemma} 

\begin{proof}
Choose a disjoint collection of $\eta$-balls in 
$ S= \rm{CH}(E) \cap W$ and note that there are $O(n)$ such 
balls since $S$ has hyperbolic area  $O(n)$ (it is a convex
hyperbolic polygon with $O(n)$ sides, hence has a triangulation into
$O(n)$ hyperbolic triangles, and every hyperbolic triangle has 
hyperbolic area $\leq \pi$). 

Every pentagon used in the proof of Lemma \ref{build intervals} 
is within a bounded hyperbolic distance $D$  
of one of the chosen $\eta$-balls, so only $O(1)$ pentagons can 
be associated to any one ball (they are disjoint, have a fixed 
area and all lie in a ball of fixed radius, hence fixed area).
 Thus the total number of 
pentagons used is $O(n)$. Every Carleson quadrilateral   shares an edge
with a pentagon and  every Carleson triangle shares an edge with 
a quadrilateral, so the number of these regions is also 
$O(n)$.
\end{proof}


\section{Meshing the pentagons} \label{pent mesh sec}

In the last section we subdivided the unit disk into hyperbolic 
pentagons, quadrilaterals and triangles. Next we 
want to mesh each of these regions  into  quadrilaterals 
with angles in the interval $[60^\circ, 120^\circ]$.
 Moreover, along 
common edges of the  regions, the vertices  of the meshes
must match up correctly.  

For each type of  region, we will produce a mesh by quadrilaterals that have 
circular arc boundaries and angles within a given range.  In most cases
the boundary arcs lie on circles with radius comparable to the region, 
and the quadrilaterals will be much smaller, about $1/N$ as large, for 
a large $N$. If we replace the circular arc edges by line segments, 
the angles change by only $O(1/N)$, which  still gives angles in 
the desired range. The only exceptions will be certain parts of the 
mesh of the Carleson triangles, that will require a separate argument 
to show the ``snap-to-a-line'' angles are still  between $60^\circ $ 
and $120^\circ$. 

As before,  $L$ denotes  the sidelength of a hyperbolic 
right pentagon.

\begin{lemma} \label{pentagon mesh} 
For sufficiently large  integers
 $N>0$ the following holds.
Suppose $P$ is a hyperbolic right pentagon. Then  there is 
mesh of $P$ into hyperbolic quadrilaterals with angles between 
$72^\circ$ and  $108^\circ$. The mesh divides each side 
of the pentagon into $N$ segments of length $L/N$. Each quadrilateral 
$Q$ in the mesh has hyperbolic diameter $O(1/N)$ 
and satisfies $\diam(Q) =O( \frac 1 N \cdot \diam(P))$ in 
the Euclidean metric. 
Replacing the edges of $Q$ by line segments changes angles by only 
$O(1/N)$.
\end{lemma} 

\begin{proof} 
Connect the center $c$ of the pentagon by hyperbolic
geodesics  to the (hyperbolic) center of each 
edge.  This divides the pentagon into five 
quadrilaterals each of which has 3 right angles 
and an angle of $ 72^\circ $ at the 
center. Consider one of these quadrilaterals $Q$ 
with sides $S_1, S_2, S_3, S_4$ where 
 $S_1$, $S_2$  each connects  the center of the pentagon 
to midpoints of adjacent sides.
Then $ S_3, S_4$  are each  half of  a side of the pentagon 
adjacent at a vertex $v$, with $S_3$ opposite $S_1$ and $S_4$ 
opposite $S_2$ (Figure \ref{penta-defn}).

\begin{figure}[htbp]
\centerline{ 
\includegraphics[height=1.75in]{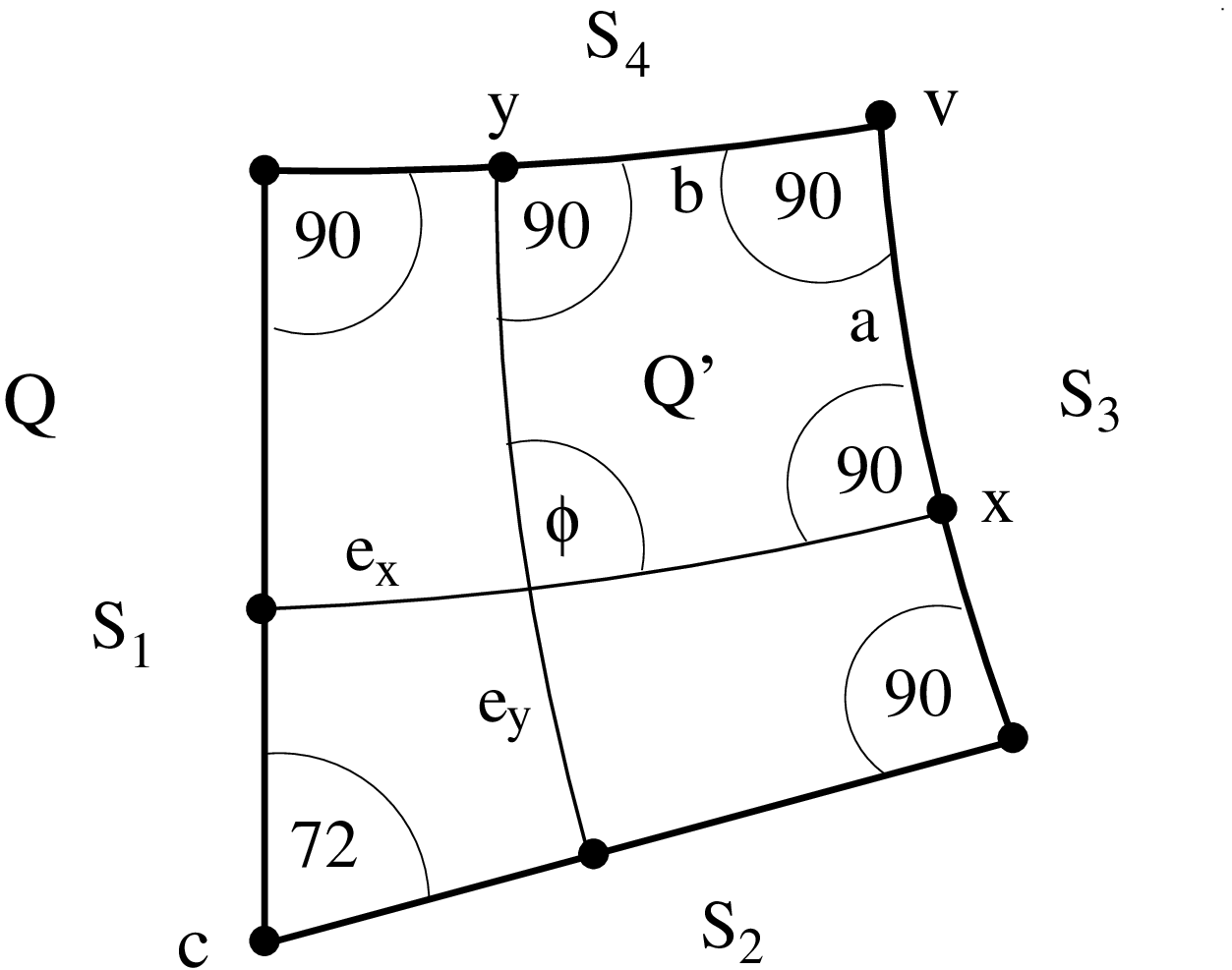}
 }
\caption{ \label{penta-defn} 
Definitions used in the mesh of a hyperbolic right pentagon. 
Each pentagon is divided into five quadrilaterals as shown.
}
\end{figure}
 
Place a point $x$  along $S_3$ and  let $e_x$ 
be the geodesic segment from $S_1$ to $S_3$ that 
meets $S_3$ at $x$ and makes a $90^\circ $ angle
with $S_3$.  Similarly define a segment 
$f_y$ that joints $y \in S_4$ to $S_2$.  
We claim that the segments cross at an angle 
(labeled $\phi$ in Figure \ref{penta-defn}) 
that is between $72^\circ $ and $90^\circ$.  The 
two segments $e_x$, $f_y$ divide $Q$ into 
four quadrilaterals, one of which contains the 
vertex $v$.  This subquadrilateral, $Q'$ is a 
Lambert quadrilateral, i.e., bounded by four 
hyperbolic geodesic segments and having 3 
right angles. The one non-right angle, $\phi$, 
is a function of the hyperbolic lengths of the
two opposite  sides (in this case a function of 
$a=\rho(x,v)$ and $b=\rho(y,v)$), 
$$ \cos(\phi) = \sinh a \sinh b .$$
See Theorem 7.17.1 of \cite{Beardon}.
Clearly, $\phi$ decreases as either $a$ or $b$ 
increase. For $a$ and $b$ close to zero we have 
$\phi \approx 90^\circ$ and when $a,b$ take their maximum value 
($a=b $ is the hyperbolic length of $S_3$) we 
get $Q' = Q$ and $\phi= 72^\circ$.
Thus $\phi $ takes values between $72^\circ$ and 
$90^\circ$, as claimed.

To define a mesh of $Q$, take $N$ equally spaced  points 
$\{ x_k\} \subset S_3$ and $\{ y_k \} \subset S_4$ and 
take the union of segments $e_{x_k}, f_{y_k}$.  This divides
$Q$ into quadrilaterals with geodesic boundaries
and angles between $72^\circ$ and $108^\circ$. 
  Doing this for each of the five quadrilaterals that 
make up the hyperbolic right pentagon gives a mesh of the 
pentagon. The remaining claims are easy to verify.
 See Figure \ref{mesh-pent}.

\begin{figure}[htbp]
\centerline{ 
\includegraphics[height=1.5in]{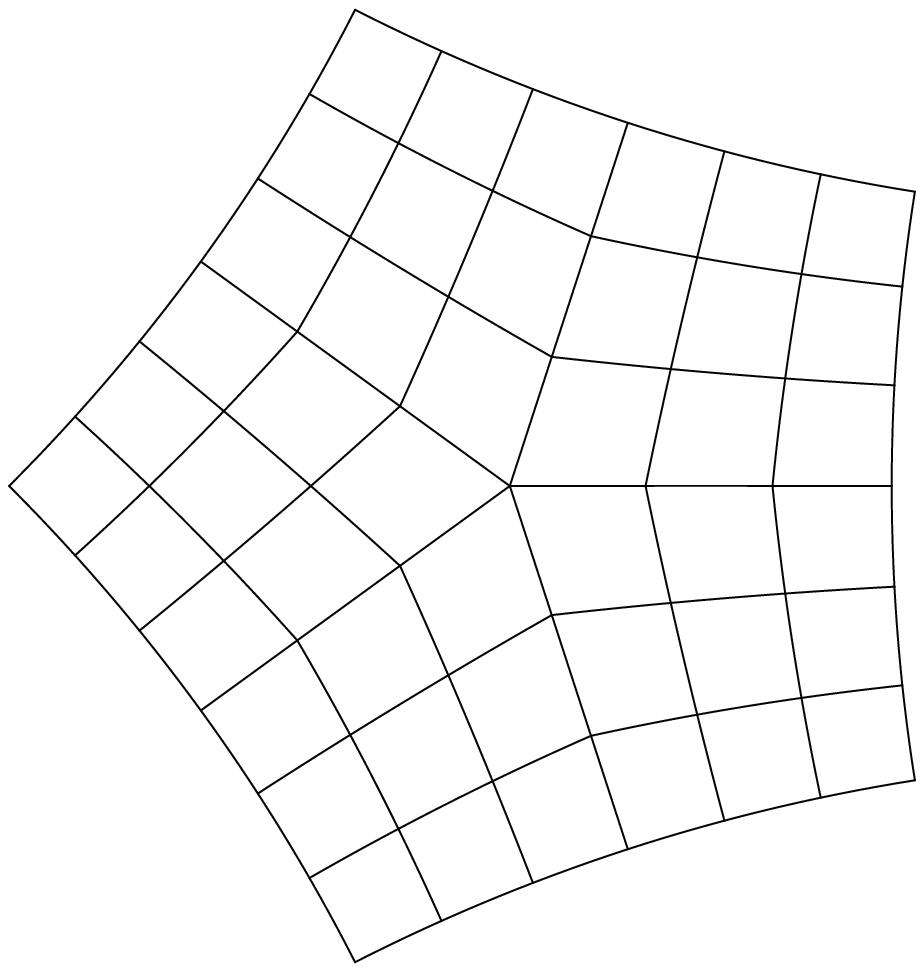}
$\hphantom{xxxx}$
\includegraphics[height=1.5in]{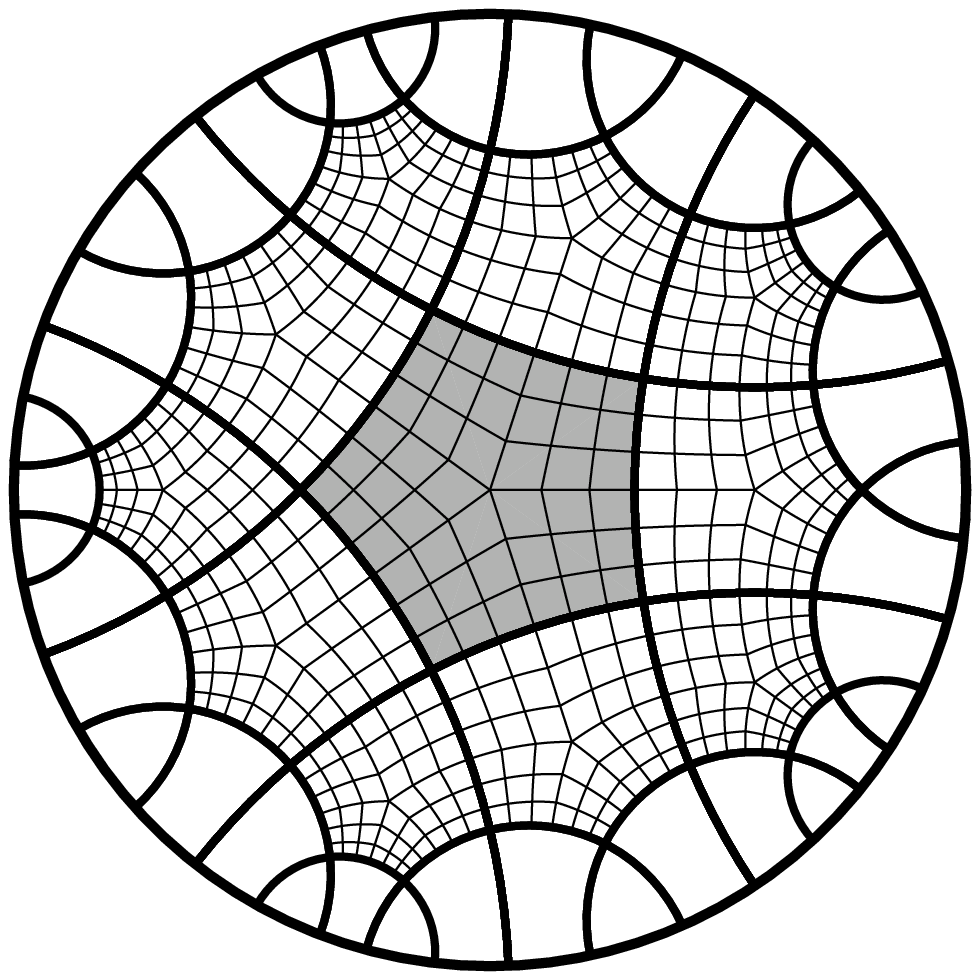}
 }
\caption{ \label{mesh-pent}
A  quadrilateral mesh of a single  pentagon and the mesh 
on 11 adjoining pentagons. 
Because vertices are evenly spaced in the hyperbolic metric,
meshing of adjacent pentagons match up.
}
\end{figure}
\end{proof}  


\section{Meshing the  quadrilaterals } \label{meshing quads}

\begin{lemma} \label{quad mesh} 
For sufficiently large integers $N$ the following holds. 
Suppose $\{ d_1 < d_2< \dots < d_M\}$ satisfy $|d_k-d_{k+1}|
\leq 1/N$ for $k=1, \dots, M-1$, $d_1< 1/N$, $ d_M > N$  and suppose 
$R$ is a  right Carleson quadrilateral. Then  there is 
mesh of $R$ into hyperbolic quadrilaterals with angles between 
$90^\circ -O(\frac 1N)$ and  $90^\circ+O(\frac 1 N)$. The mesh divides  
the unique finite (hyperbolic) length side  of $R$ into $N$ 
segments of length $L/N$.  Each infinite length side of $R$  has vertices 
exactly at the points that are hyperbolic distance $d_k$, 
$k=1, \dots, m$ from the finite length side. 
If the base of $R$ has length $\leq \pi$, then 
each element $Q$ of the mesh 
satisfies $\diam(Q) = O(\frac 1 N  \cdot \diam(R))$  in the Euclidean 
metric.
Replacing the edges of $Q$ by lines segments changes angles by at most 
$O(1/N)$.
\end{lemma}

We need a simple preliminary result.

\begin{lemma} \label{perp foliation}
Suppose $Q$ is a right circular quadrilateral, i.e., is bounded by 
four circular arcs and all four interior angles are $90^\circ$. Then $Q$ 
has two orthogonal foliations by circular arcs. Every leaf of 
both foliations is perpendicular to the boundary at both of its 
endpoints. 
\end{lemma}

\begin{proof}
To see this, take two opposite
sides. Each lies on a circle and  these circles either intersect
in 0, 1 or 2 points or are the same circle.
In the first case
we can conjugate by a M{\"o}bius transformation so both disks are
centered at $0$. Then the two other sides must map to radial segments
and the foliations are  as claimed.
 If the circles
intersect in two points, we can assume these points are $0$ and $\infty$
so  the circles are both  lines passing
through $0$ and again the foliations are radial rays and circles centered
at $0$. If the opposite sides  belong to the same circle,
we can  conjugate it to be the real line, with the two sides being arcs
symmetric with respect to the origin. Then the other two sides
must be circular arcs centered at $0$ and the two foliations are as before.
The last, and exceptional, case is
 if the two circles intersect in one point. Then we  can conjugate this
point to infinity and the intersecting sides to
two parallel lines. The other two sides must map to perpendicular
segments and the region is  foliated by perpendicular straight lines.
See Figure \ref{circ-quad}.
\end{proof}

\begin{figure}[htbp]
\centerline{
\includegraphics[height=1.25in]{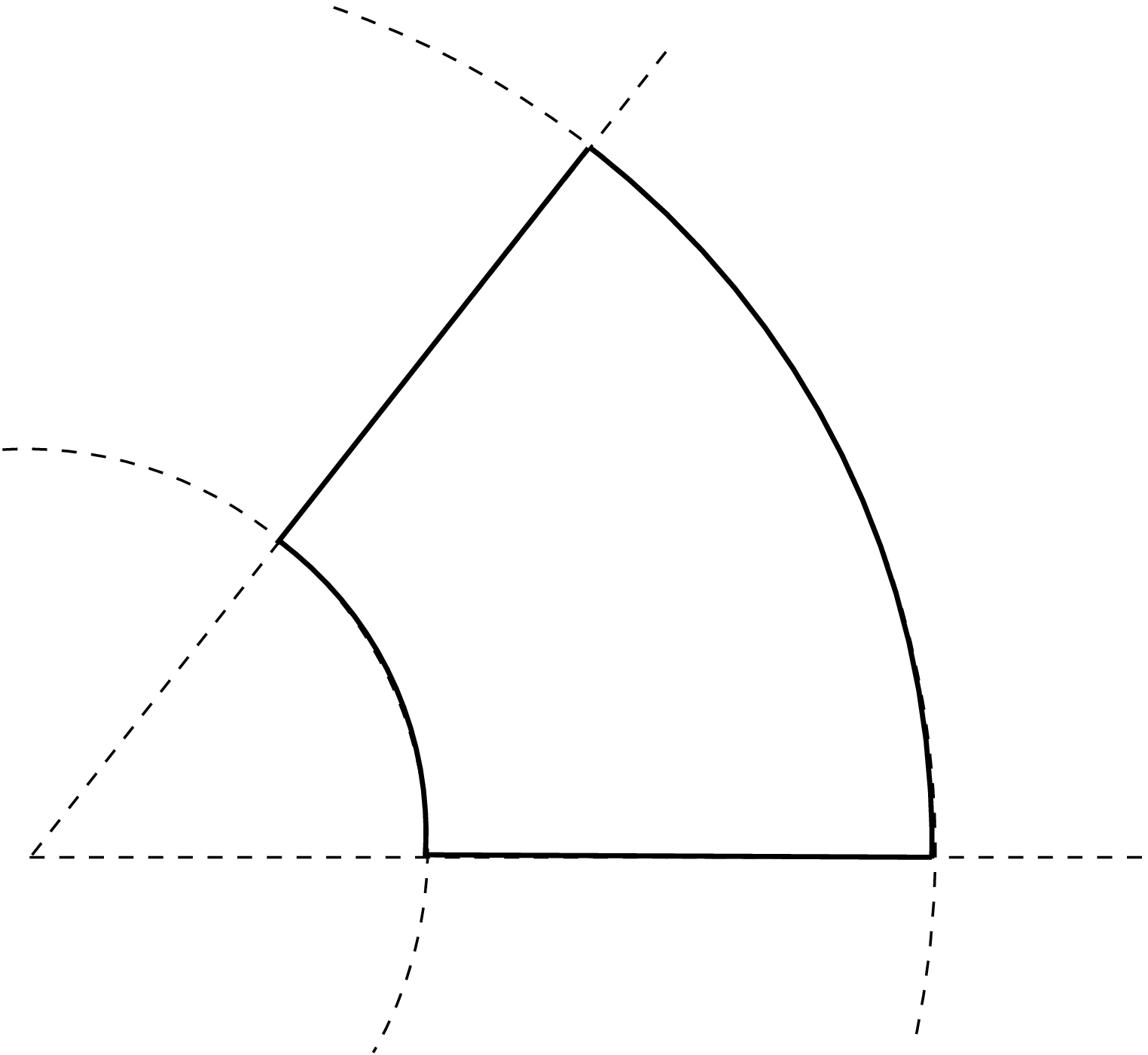}
$\hphantom{x}$
\includegraphics[height=1.25in]{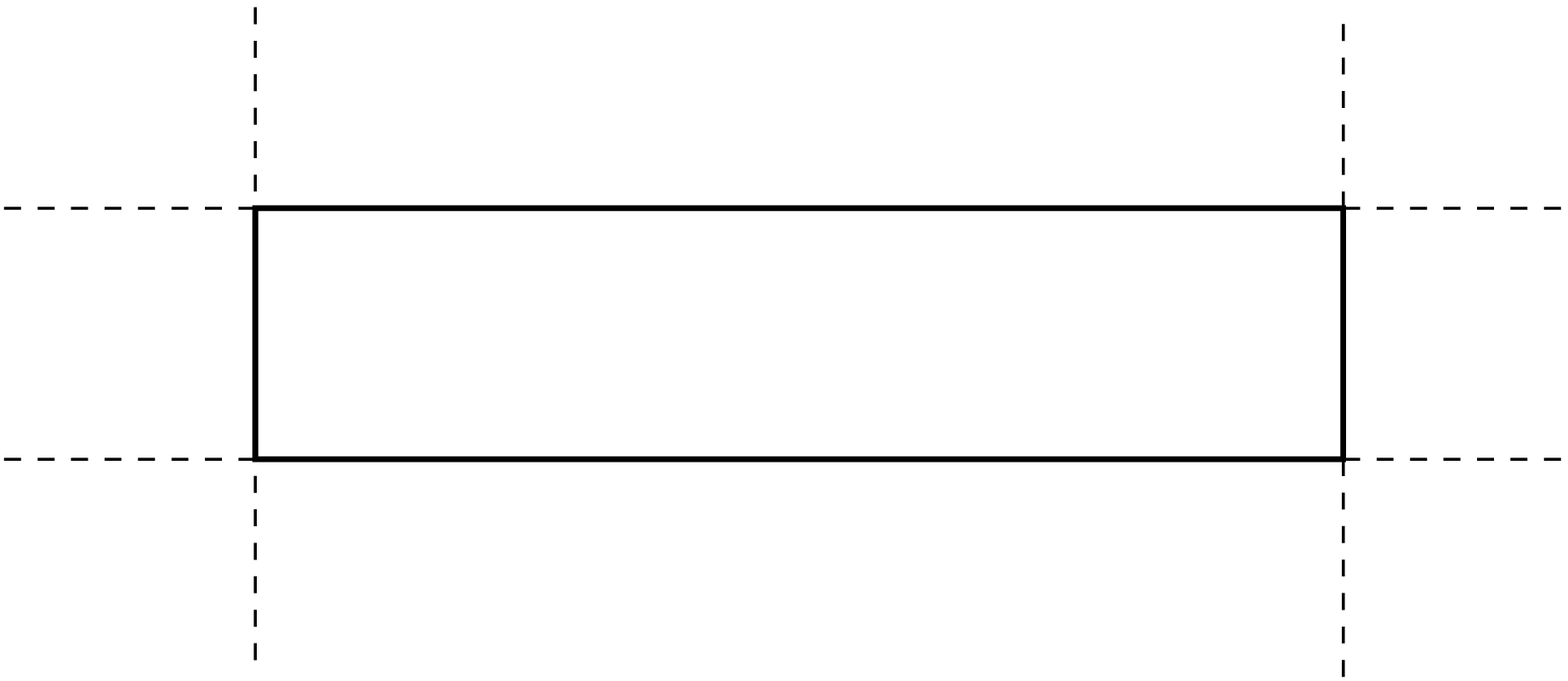}
}
\caption{ \label{circ-quad}
Any right circular quadrilateral is M{\"o}bius equivalent to
one of these cases and hence has an orthogonal foliations by circular
arcs.
}
\end{figure}

\begin{proof} [Proof of Lemma \ref{quad mesh}]
The two sides of $R$ that lie in $\disk$ but have infinite 
hyperbolic length are geodesic rays that are 
both perpendicular to the  geodesic containing the 
top edge of $R$. 
 Hence they are subarcs of non-intersecting circles
 (to see this, isometrically 
map $\disk \to \uhp$ so the top edge  maps to a vertical 
segment and the geodesic rays   map to arcs of concentric circles). 
The foliations provided by the previous lemma consist of 
(1) hyperbolic geodesics that are 
perpendicular to the top edge of $R$ (the unique finite 
length side) and  (2) subarcs of  circles that 
all pass through $a,b$ (the endpoints of the hyperbolic 
geodesic that contains the top edge of $R$).  We call 
these the vertical and horizontal foliations respectively.

To prove the lemma, we simply subdivide the edges of $R$ 
as described and  take the foliation leaves with these 
endpoints.  The only point that needs to be checked is
that points on the two infinite length sides of $R$ 
that are the same hyperbolic distance from the top 
edge lie on the same horizontal foliation leaf. 
However, any two  horizontal leaves are  equidistant from 
each other in the hyperbolic metric 
(to see this, map the vertices $a,b$ to $0,\infty$ by an 
isometry $\disk \to \uhp$ and these leaves become 
rays, and the claim is obvious since dilation is 
an isometry on $\uhp$).
 Since the top edge is a horizontal leaf, we are done. 
 See Figure \ref{fill-quad}.

\begin{figure}[htbp]
\centerline{ 
	\includegraphics[height=1.75in]{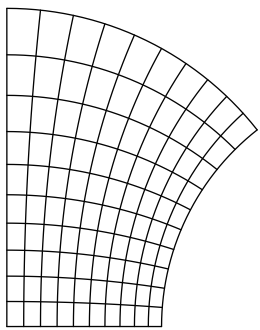}
	}
\caption{ \label{fill-quad} 
A quadrilateral mesh of a Carleson quadrilateral.  ``Horizontal'' edges
lie on circles that pass through the same two points on the boundary (the 
endpoints of the geodesic contain the top edge).
``Vertical'' edges are hyperbolic geodesics perpendicular to the 
top edge.
}
\end{figure}

\end{proof} 

\section{Meshing   the triangles } \label{triangles}

Unlike our meshes of the Carleson  quadrilaterals 
and right pentagons, our mesh of the Carleson 
triangles will use the full interval of angles 
$[60^\circ, 120^\circ]$. 
This is easy to do if we just want to mesh by 
quadrilaterals with circular arc sides. However, 
we will want to conformally  map our mesh in $\disk$ to $\Omega$
and then replace the curved  edges  in the image 
by straight line segments.  This  can change the angles 
slightly, so we would end up with angles in $[60^\circ - \epsilon, 
120^\circ + \epsilon]$ (where $\epsilon$  depends on 
the ratio
between the diameters of our mesh elements and the 
diameter of $T$).
To get the sharp result,   we will have to 
 be careful how we use angles near 
$60^\circ$ and $120^\circ$. To simplify matters, it 
will be enough to simply consider one special Carleson 
triangle $T$ in the upper half-plane model with vertices 
at $-1,1, i/(\sqrt{2}-1)$. The mesh for any other 
triangle will be obtained as a M{\"o}bius image of the 
mesh we construct on this triangle.

The triangle $T$ has one vertex in $\uhp$, and we 
refer to this as the ``top point''. Adjacent to the 
top point are two sides that we call the ``left'' and 
``right'' sides. 
Inside  $T$  we will construct 
an ``inner triangle'' $T_i \subset T$.
The vertices of $T_i$ form an ordinary 
equilateral Euclidean triangle, but the edges of $T_i$
itself are 
circular arcs meeting at three interior angles of $90^\circ$,
and $T_i$ is uniquely determined by this.

\begin{lemma} \label{triangle mesh} 
The following holds for all sufficiently large integers 
$N$.  There is a sequence $d_1 < d_2 < \dots < d_M$ with 
$|d_k - d_{k+1}| \leq 1/N$ for $k=1, \dots, M-1$  and a 
mesh of $T$ into  hyperbolic
quadrilaterals with angles between $60^\circ$ 
and $120^\circ$ so that the vertices along the left and 
right edges of $T$ occur exactly at the points distance 
$d_k$, $k=1, \dots, m$ from the top point. 
Every quadrilateral $Q$ in the mesh satisfies
$\diam(Q) = O(\frac 1 N  \cdot \diam(T)) $.  The triangle 
$T$ contains a symmetric right circular triangle $T_i 
\subset T$ so that  outside $T_i$, only angles 
in $[90^\circ -O(\frac 1N) , 90^\circ + O(\frac 1N)]$ are used.
The triangle $T_i$ may be chosen as small as wish compared
to $T$.
Replacing edges by straight line segments gives angles between
$60^\circ$ and $120^\circ$.
\end{lemma}

\begin{figure}[htbp]
\centerline{
	\includegraphics[height=2in]{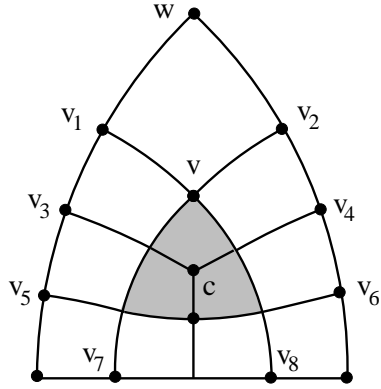}
	}
\caption{ \label{inner-outer} 
The outer triangle $T$ is a Carleson triangle in the upper half-plane 
with top point $w = i/(\sqrt{2}-1)$.
Its interior is  divided into an inner triangle  $T_i$ (shaded) 
with top point $v$ and nine surrounding right circular quadrilaterals.
The points $v_1, v_2$ are  equidistant from $w$
in the hyperbolic metric. The left  and right sides of $T_i$ 
are geodesic segments and extend to hit $\reals$ as points 
$v_7, v_8$. The Carleson triangle with vertices $v, v_7, v_8$ 
is denoted $T_e$.
}
\end{figure}

 The inner triangle $T_i$ it is divided 
into three quadrilaterals by connecting the center of the triangle
 to the midpoint of 
each edge by a straight line.
The vertices of $T_i$ and the midpoints of its edges 
are connected to points on $\partial T$ by 
circular arcs that are perpendicular to both the boundaries
of  $T$ and $T_i$ at the points where they meet.
See Figure \ref{inner-outer}. 
We mesh each of the nine resulting quadrilaterals using 
the foliations given in Lemma \ref{perp foliation}, starting 
at the left and right sides of $T$ at the points given 
by $\{ d_k\}$. We assume that this collection contains the 
distances $\rho(w,v_1), \rho(w, v_3), \rho(w, v_5)$.
 When a leaf ends we continue it in the next 
quadrilateral (assume we know how to do this for the inner
triangle and that the foliation there is symmetric).
 The path continues until it either it hits $c$ (the center of
the inner triangle),  
hits  $[-1,1]$ (the base of $T$) or hits the opposite side of $T$.
 In the 
latter case, symmetry implies the path ends at a point
the same distance from the top point as its starting point.


The choice of inner triangle $T_i$  depends only on the choice of 
its top point. This lies on  the positive imaginary axis,
and $T_i$ is chosen to be symmetric with respect to this 
line. The diameter of $T_i$ is scaled so that the left and right 
edges of $T_i$ are hyperbolic geodesic segments (if the top point 
has height $h$ above $0$, the three vertices of $T_i$ should 
form an equilateral triangle of sidelength $h(\sqrt{3}-1)$; 
see Figure \ref{ScaleInner}).
 Since any point 
between the top point of $T$  and the origin can be used, the inner
triangle can be as small as we wish.

\begin{figure}[htbp]
\centerline{
	\includegraphics[height=1.5in]{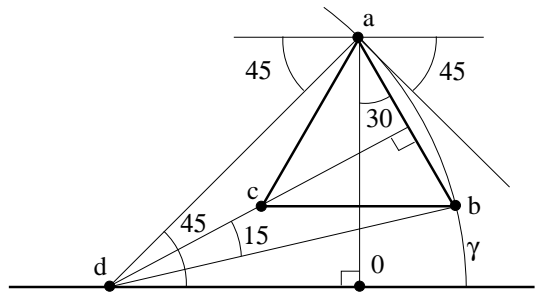}
	}
\caption{ \label{ScaleInner} 
How to scale the inner triangle.
Suppose $a$ is height $1$ above the real axis and $a,b$ lie 
on a geodesic  $\gamma$ centered at $d$ that makes a $45^\circ$ angle 
with the horizontal at $a$. The $\Delta a0d$ is isosceles with base 
angles $45^\circ$, so $|ad| = |bd|=\sqrt{2}$.   The line $da$ is
perpendicular to $\gamma$, so $\Delta dab$ is isosceles.
 Thus $|ab| = 2 |bd| \sin(15^\circ) = \sqrt{3}-1$.
}
\end{figure}


We define three foliations on this triangle $T_i$. For each vertex $v$, 
reflect $v$ through the circular arc on the opposite side 
to define a point $v^*$ and foliate $T_i$ by arcs that lie on 
circles passing through both $v$ and $v^*$.
 Note that each foliation leaf passes through one 
of the vertices of $T_i$ and is perpendicular to the opposite side.
See Figure \ref{3-full-foliations}.
 The center of the triangle  can be connected 
to the midpoint of  each side  by a foliation leaf that is 
a straight line,
 dividing $T$ into three 
quadrilaterals. Restrict each foliation to the two quadrilaterals that are
not adjacent to the vertex it passes through.  This gives two 
foliations on each quadrilateral. See Figure \ref{3-full-foliations}.
Taking a finite set of leaves for each foliation gives a quadrilateral 
mesh of the right circular  triangle.

\begin{figure}[htbp]
\centerline{ 
\includegraphics[height=1.2in]{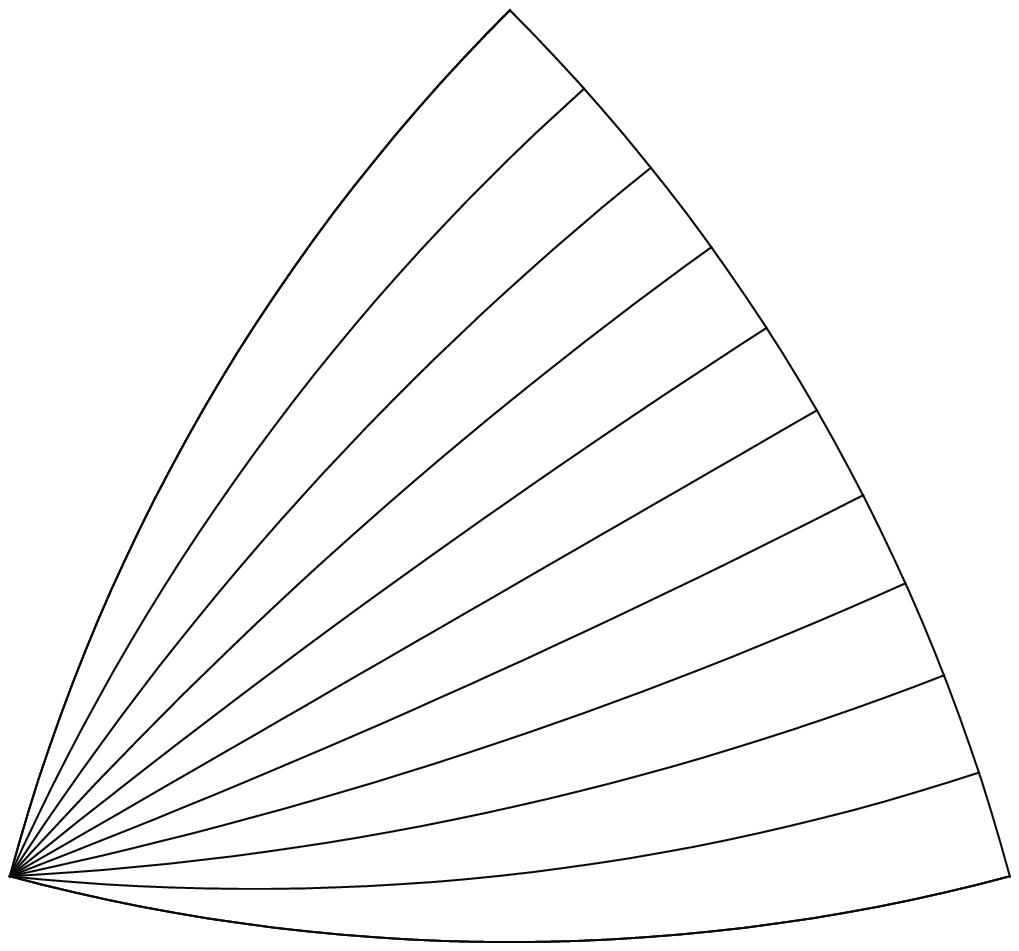}
$\hphantom{x}$
\includegraphics[height=1.2in]{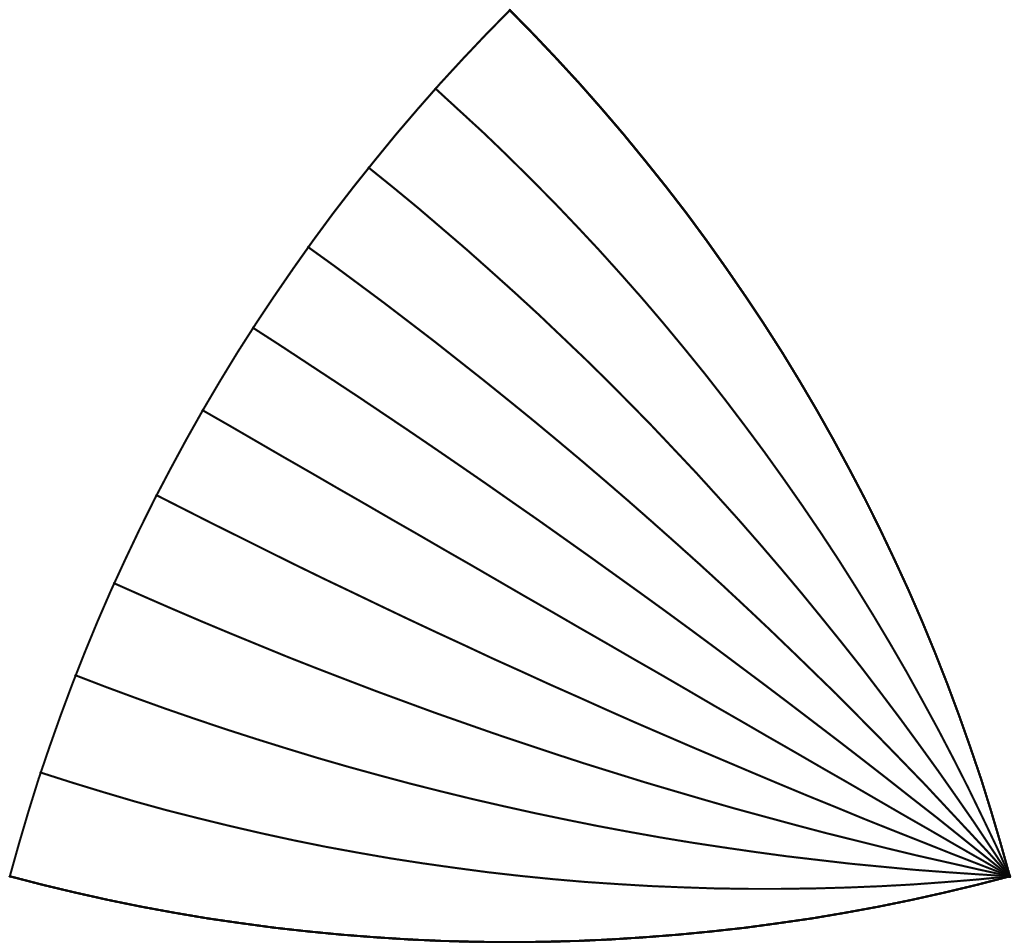}
$\hphantom{x}$
\includegraphics[height=1.2in]{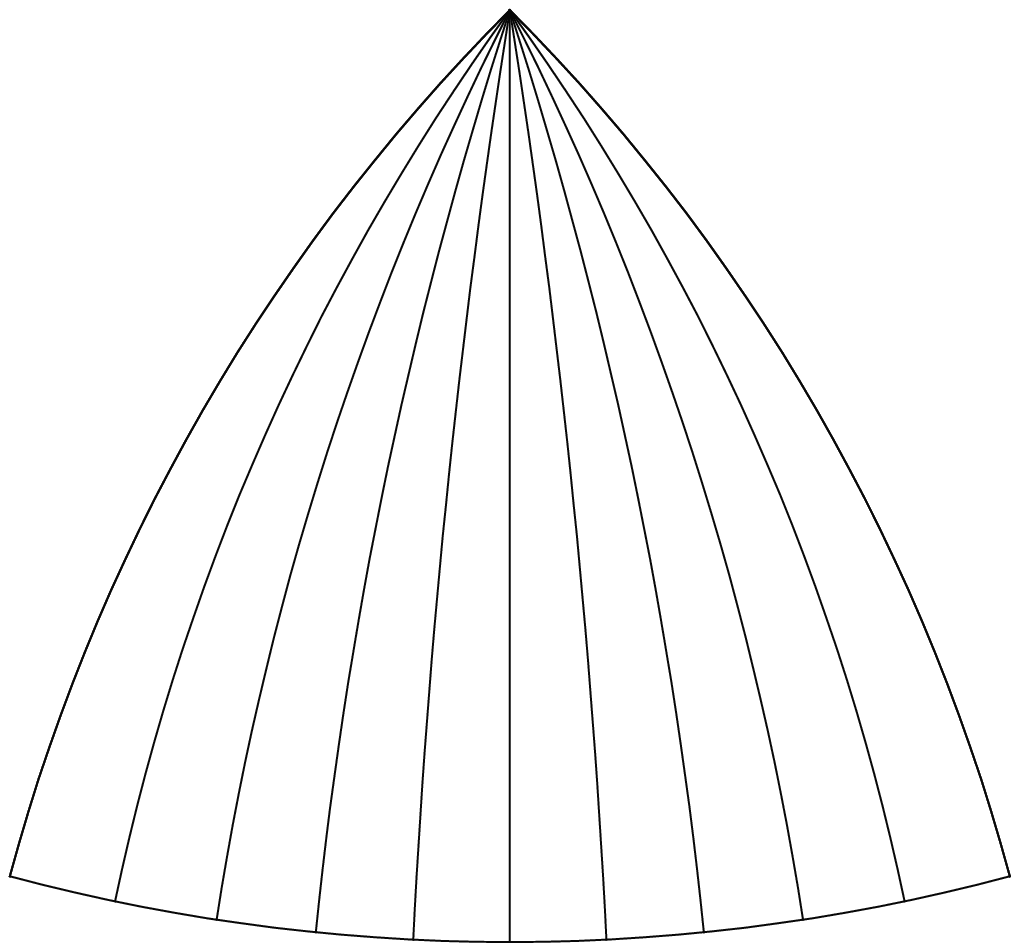}
$\hphantom{x}$
 \includegraphics[height=1.2in]{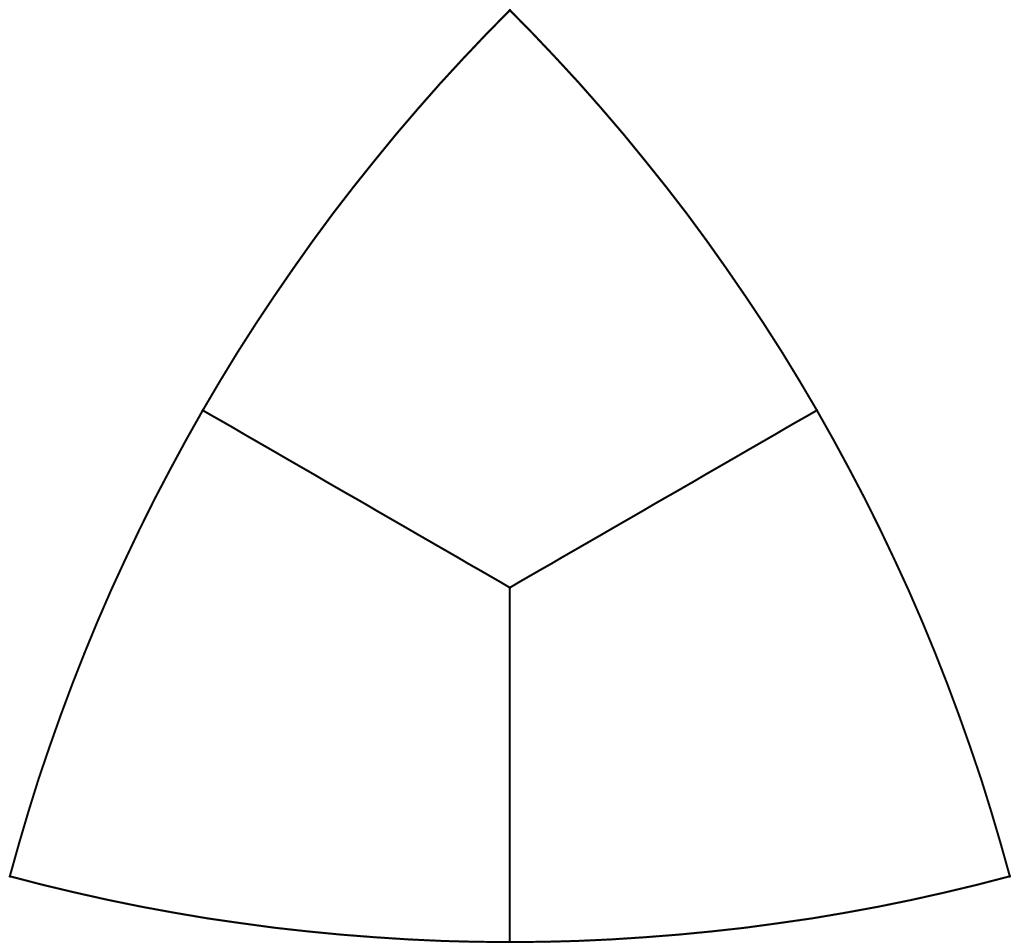}
 }
\vskip .15in
\centerline{ 
\includegraphics[height=1.2in]{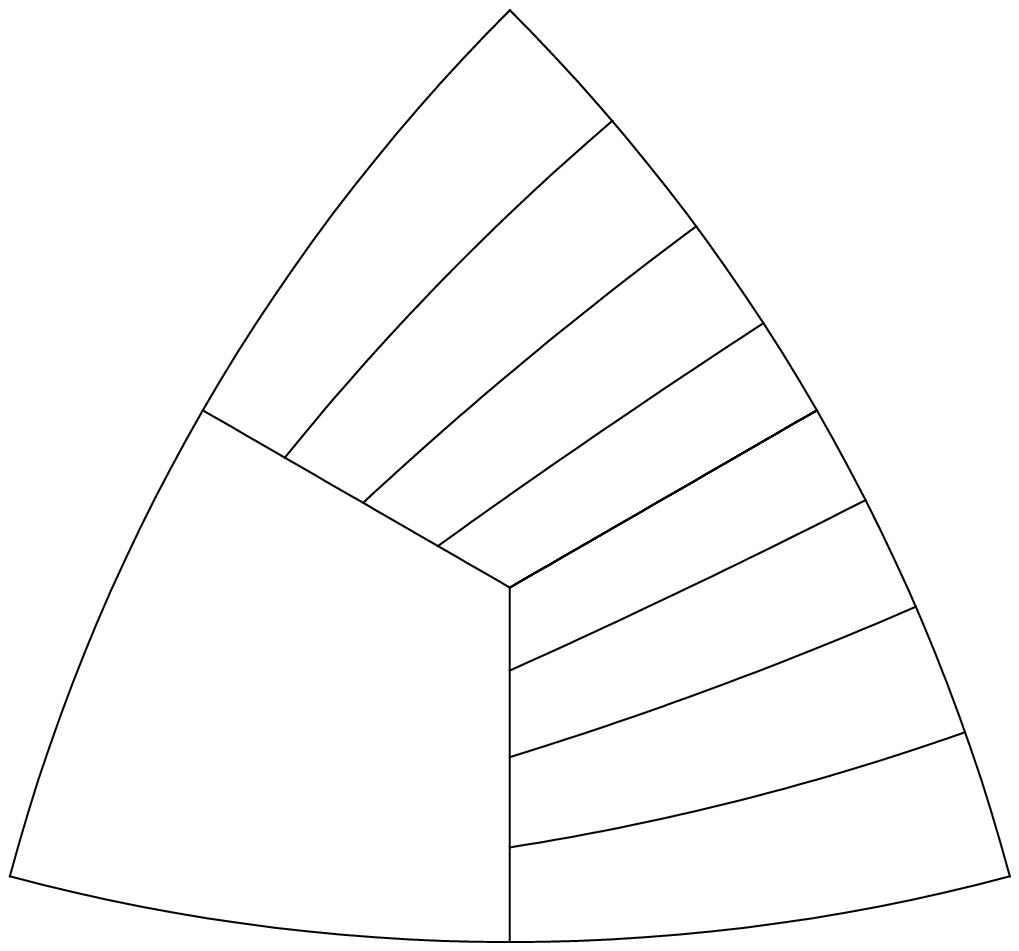}
$\hphantom{x}$
\includegraphics[height=1.2in]{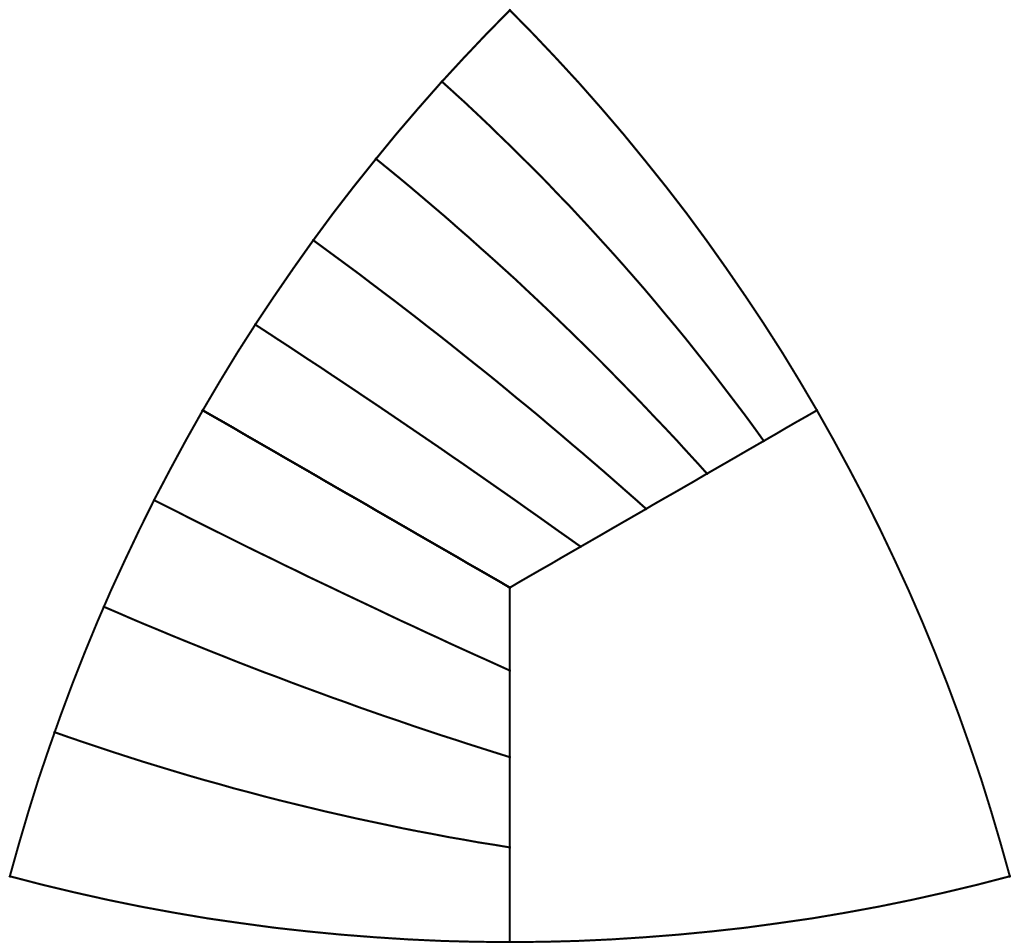}
$\hphantom{x}$
\includegraphics[height=1.2in]{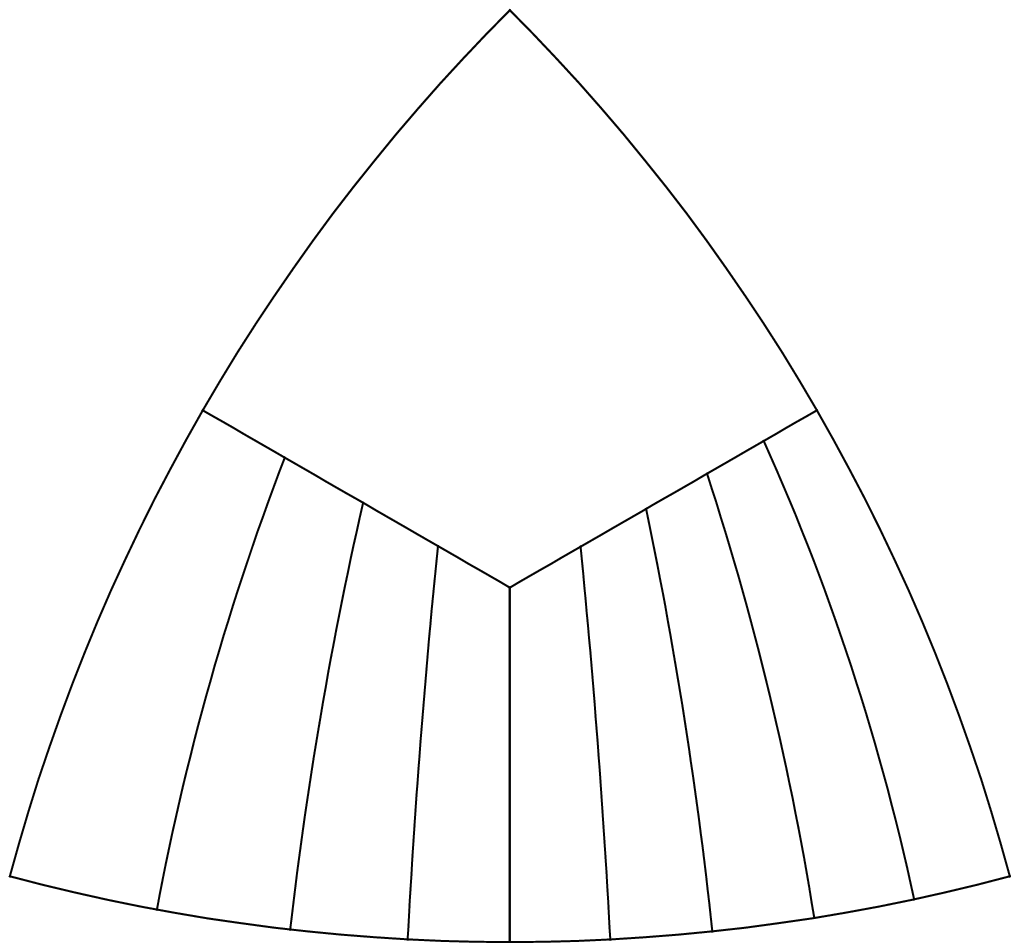}
$\hphantom{x}$
\includegraphics[height=1.2in]{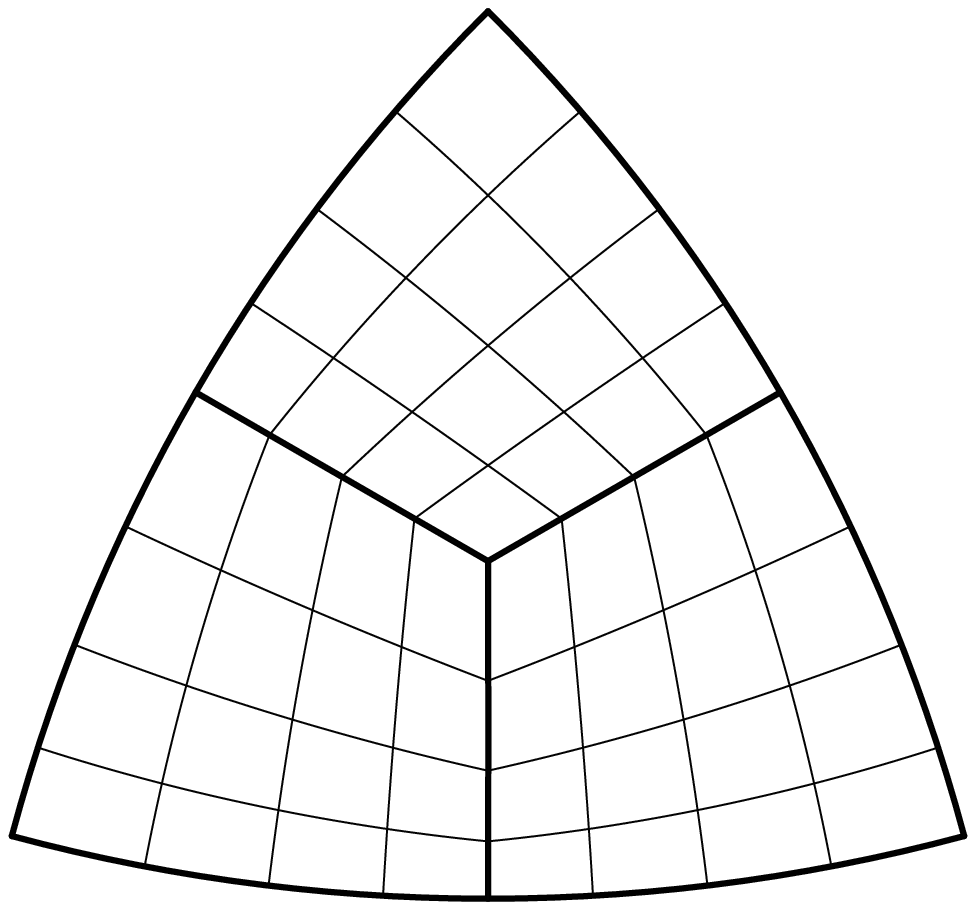}
 }
\caption{ \label{3-full-foliations} 
Three foliations of a circular right triangle. Each leaf passes
through an associated vertex and is perpendicular to the opposite side.
 Connecting the center of the triangle to the midpoints of each 
side by the straight leaf divides $T_i$ into three 
quadrilaterals.
We then restrict each foliation to two of the quadrilaterals 
as shown, and leaves of the union give the mesh edges. 
}
\end{figure}

Combining this foliation of the inner triangle with the foliations of 
the surrounding quadrilaterals and choosing starting points along the 
left and right sides of $T$ as described earlier gives the desired 
mesh of $T$. See Figure \ref{full-mesh-tri}. 
 The only part of the lemma left to 
prove is the claim that the angle are in the desired interval when 
replace the curved edges by straight segments. 

\begin{figure} 
\centerline{ 
\includegraphics[height=2.25in]{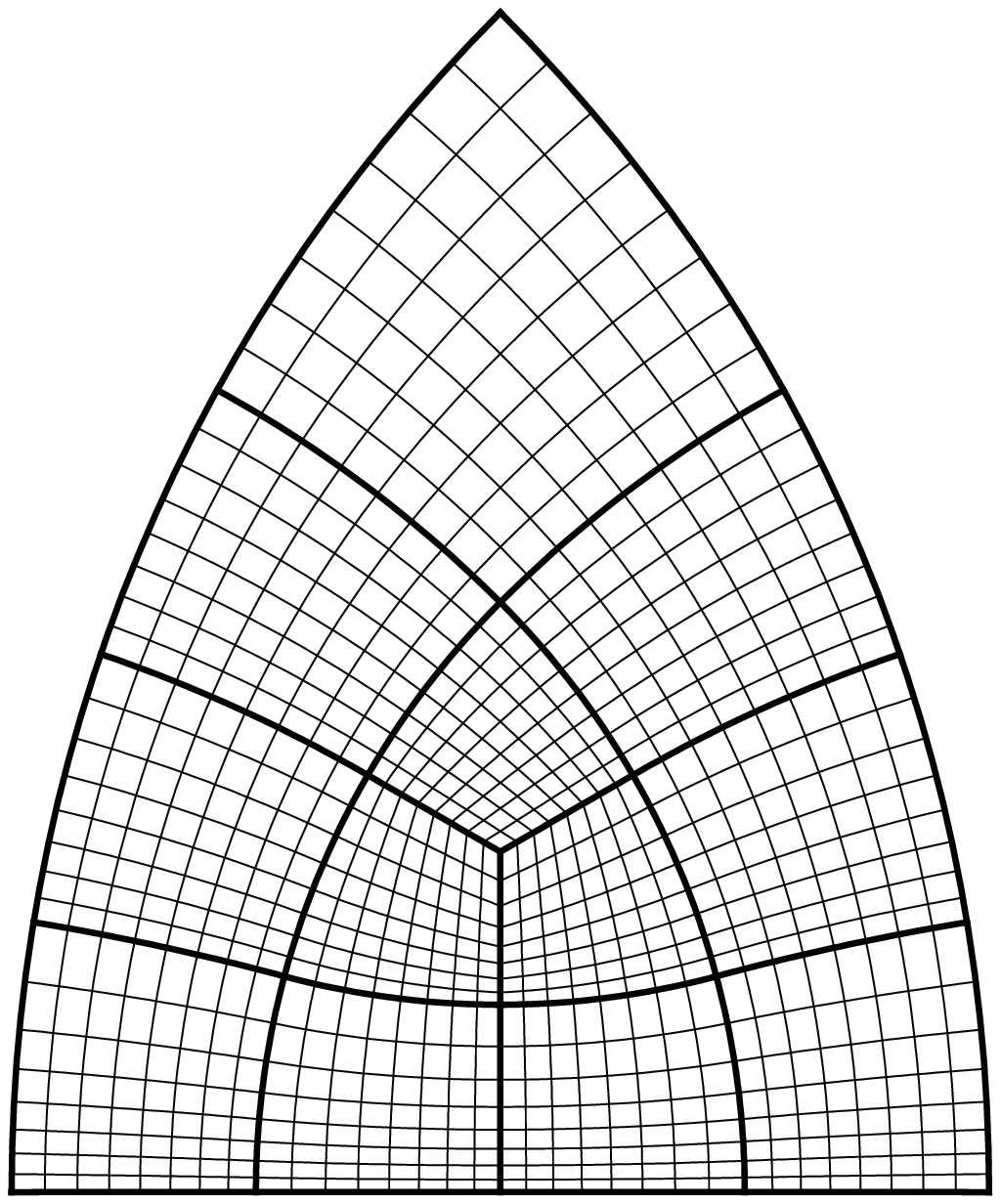}
$\hphantom{xxxx}$
\includegraphics[height=2.25in]{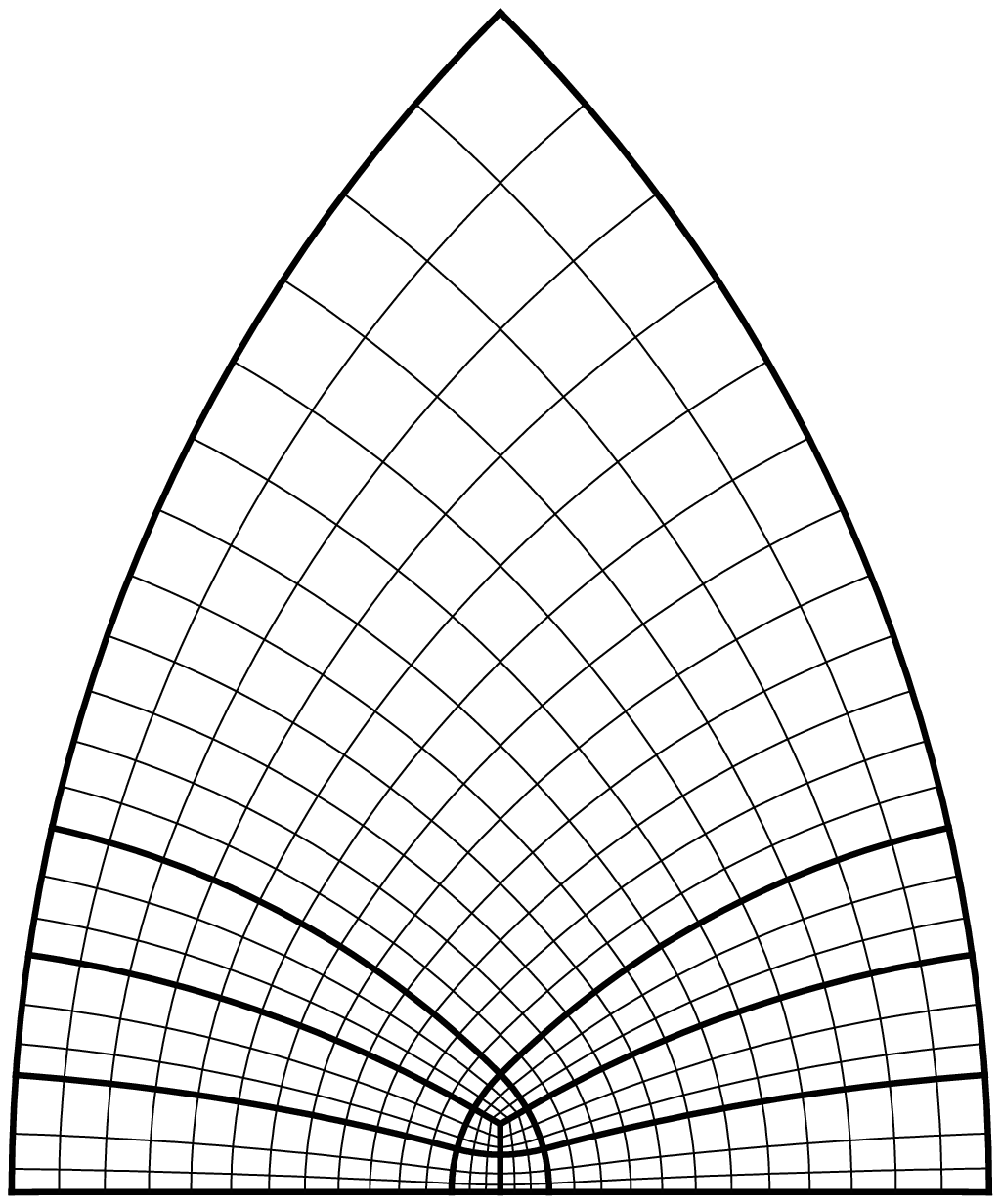}
 }
\caption{ \label{full-mesh-tri} 
The mesh of a Carleson triangle for two  different 
positions of the inner triangle. 
}
\end{figure}

When we replace the circular arc edges in the mesh by straight line segments, 
it is not obvious that all the angles remain in 
 $[60^\circ, 120^\circ]$, but we will show that this is true.
Consider a point $z$ in one of the three quadrilaterals
and the two foliation paths $\gamma_1, \gamma_2$ 
 that connect it to the two opposite 
vertices, $v_1, v_2$ respectively. 
See Figure \ref{bounds-on-quad}. 
Let $L_1, L_2$ be the lines through the center  $c$ and the points 
$v_1, v_2$. 
If we think of the arc $\gamma_1$ as a 
graph over the line $L_1$ it is monotonically increasing as we move
away from $v_1$ and remains increasing  so as long as we stay inside the 
triangle (since $\gamma_1$ is perpendicular the the opposite
side of the triangle, the point of greatest distance from 
$L_1$ occurs outside the triangle). Thus if we translate $L_1$ to 
pass through the point $z$, we see that $\gamma_1$ stays on one side 
of this new line  up to $z$ and on the other side beyond $z$. 
Thus any chord of $\gamma_1$ in the triangle with one endpoint at 
$z$ also stays on the same side of the line as the corresponding arc
of $\gamma_1$.    Similar for $\gamma_2$ and 
$L_2$. See the right side of Figure \ref{bounds-on-quad}.
Thus if we replace foliation paths by segments, 
 at each 
vertex there will be two angles less than $120^\circ$ and two greater 
than $60^\circ$ (which are the angles formed by  $L_1$ and $L_2$).
This completes the proof of Lemma \ref{triangle mesh}.

\begin{figure}[htbp]
\centerline{ 
\includegraphics[height=1.5in]{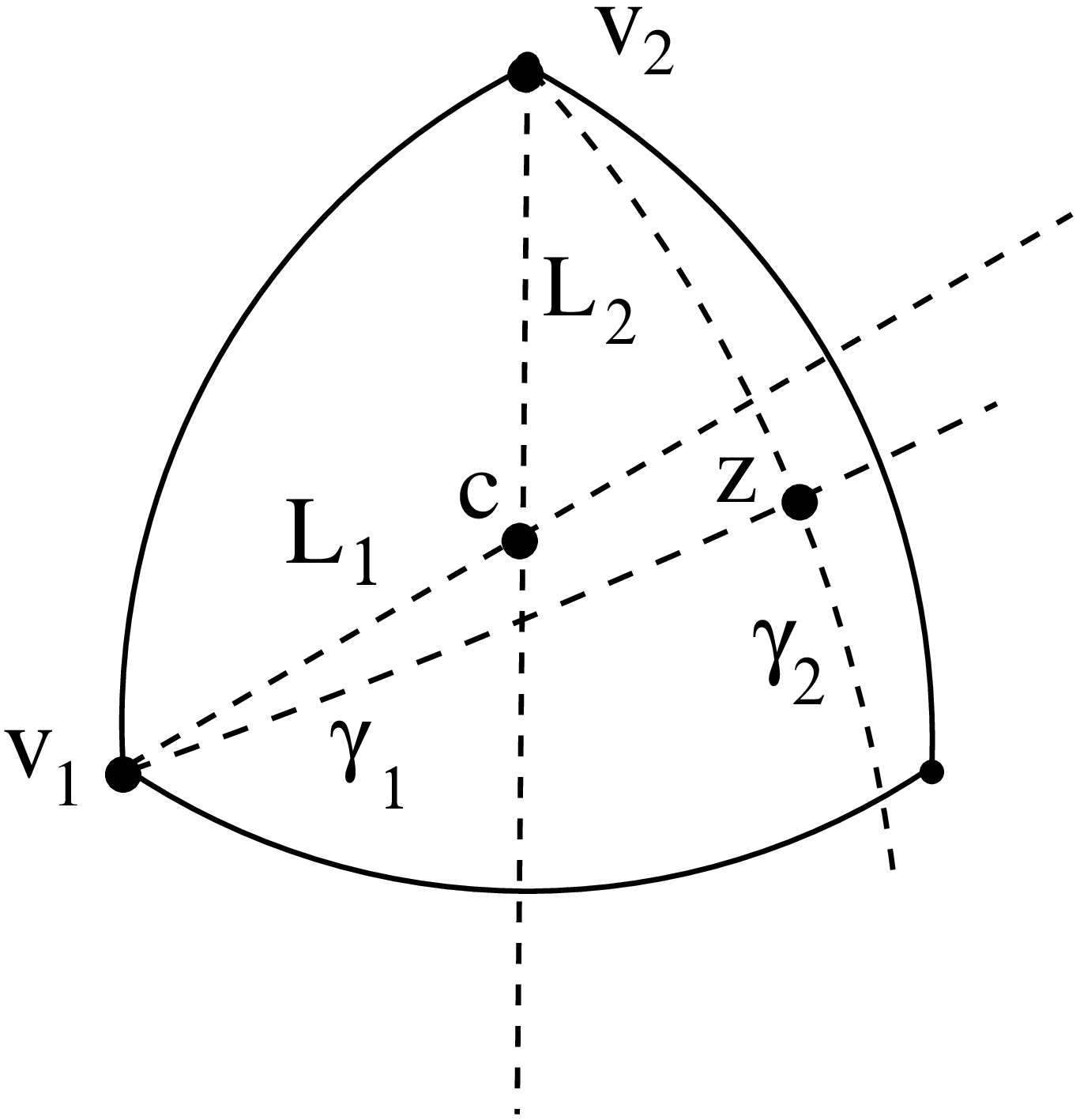}
$\hphantom{xxxxxxxxx}$
\includegraphics[height=1.5in]{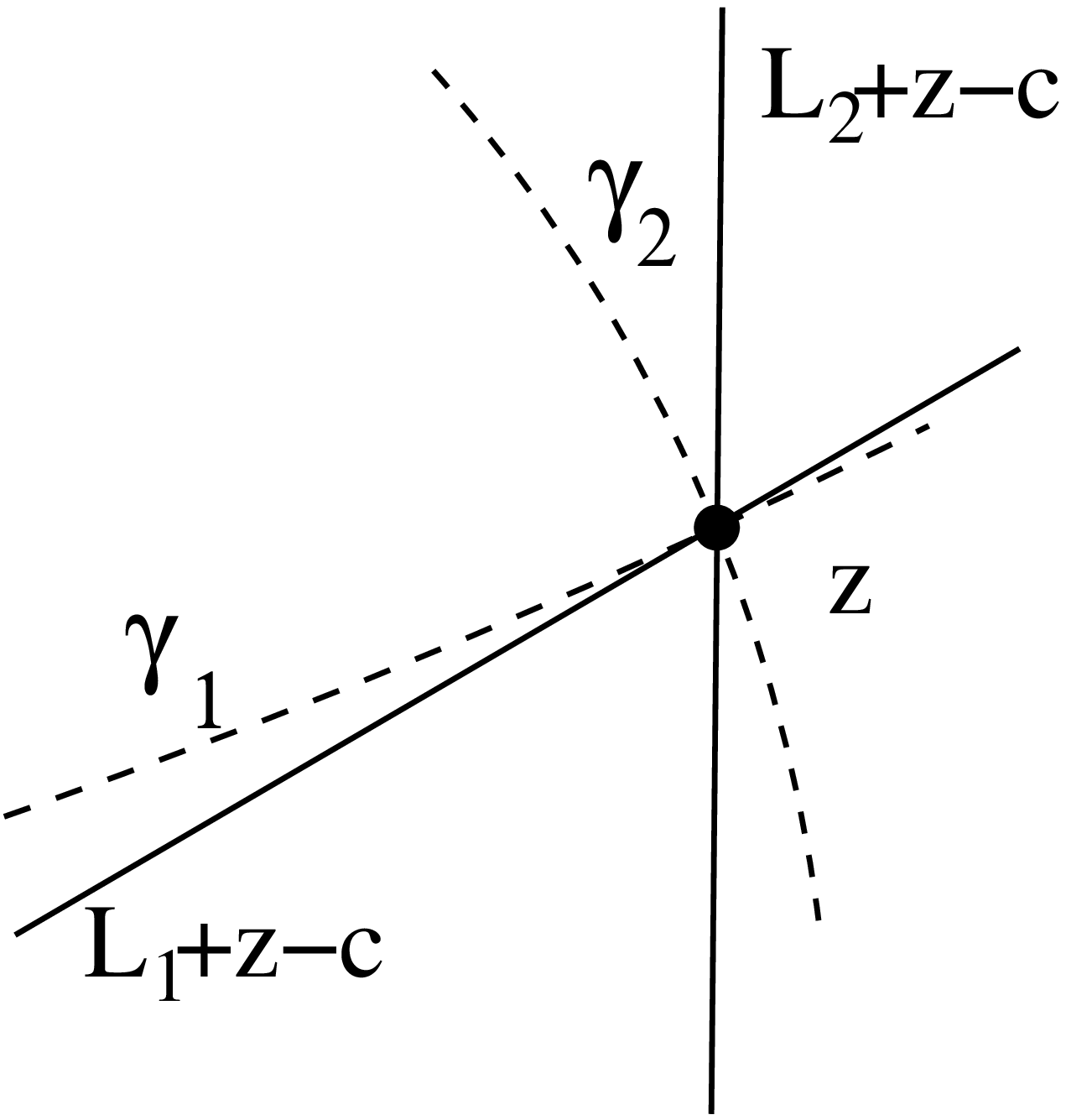}
 }
\caption{ \label{bounds-on-quad} 
If we choose any point $z$ of the equilateral triangle then 
chords of the foliation paths with endpoint $z$ form angles 
that are bounded between $60^\circ$ and $120^\circ$. 
}
\end{figure}

\section{Meshing the thick parts by conformal maps } \label{thick and thin}

The Riemann mapping theorem says that given any
simply connected planar domain $\Omega$ (other than the
whole plane) there is a 1-1, onto, holomorphic map
of the disk onto $\Omega$.  Moreover, we may map $0$
to any point of $\Omega$ and specify the
argument of the derivative at $0$.
Such a mapping is conformal, i.e., it preserves  angles
locally.
  More importantly, a conformal mapping is close to
linear on small balls with estimates that depend on the ball
but not on the mapping.
 Koebe's estimate (e.g. Cor. 4.4 of
\cite{Garnett-Marshall}) says that
if $f: \disk \to \Omega$  is  conformal 
then 
$$ \frac 14 |f'(z)| \leq  \frac {\dist( z, \partial \Omega_1) }
                       {\dist(f(z) , \partial \Omega_2) }
\leq  |f'(z)|.$$
The closely related distortion theorem  states (Equation
(4.17) of \cite{Garnett-Marshall}) that if $f$ is conformal
on the unit disk, then
$$
 \frac {1-|z|}{(1+|z|)^2}
  \leq  \frac { |f'(z)|}{|f'(0)|}
  \leq \frac {1+|z|}{(1-|z|)^3},$$
This says that on small balls $f'$ is close to constant, and 
hence that $f$ is close to linear. 
More precisely,
if $f$ is conformal on a ball $B(w,r)$ then
\begin{eqnarray} \label{close-to-linear}
 |f(z) - L(z) | \leq O(\epsilon^2 |f'(z)| r) ,
\end{eqnarray}
for $z \in B(w, \epsilon r)$,
where $L(z) = f(w) + (z-w) f'(w)$ is a Euclidean similarity.

We are particularly interested in conformal maps onto 
polygons. In this case $f$ is given by 
 the Schwarz-Christoffel formula 
$$ g(z) = A + C \int \prod_{k=1}^{n-1} (w-z_k)^{\alpha_k-1} dw , $$
where 
the interior angles  of $\Omega$  are $ 
\{ \alpha_1 \pi, \dots, \alpha_n \pi\}$ and the preimages
of the vertices are $\z= \{ z_1, \dots, z_n\}$.
 See e.g.,  \cite{DT-book}, \cite{Nehari}, \cite{DT99}.
The formula was discovered independently by
Christoffel in 1867 \cite{Chr67}  and Schwarz in 1869
\cite{Sch90}, \cite{Sch69a}. For other references
and a brief history see Section 1.2 of \cite{DT-book}.
The difficulty in using the formula is to find the 
correct parameters $\z$ for a given $\Omega$.

  For a conformal map $f$ onto a polygonal region, 
the points of the prevertex set $ \z  \subset \circle$ are the only 
singularities of $f$ on $\circle$. 
The map extends 
analytically across the complementary intervals by the 
Schwarz reflection theorem. Thus for a point $w \in \disk$, 
the map $f$ extends 
 to be conformal on the ball $B=B(w, \dist(w,E))$, 
and if $Q \subset B$ and $\diam(Q) \leq \epsilon \cdot 
\dist(Q,E)$, then there is a linear map $L$ so that 
\begin{eqnarray} \label{close-to-linear2}
 |f(z) - L(z) | \leq O(\epsilon \diam(f(Q))  ) ,
\end{eqnarray}
for $z \in Q$. In particular, the images of the vertices 
of $Q$ map to the vertices of quadrilateral whose angles 
differ by only $O(\epsilon)$ from the angles of $Q$. This is
what allows us to map our mesh via a conformal map and obtain 
a mesh with only slightly distorted angles. More precisely, 

\begin{lemma} \label{small distortion}
Suppose $f:\disk \to \Omega$ is a conformal map onto a polygonal 
domain with singular set $\z$ and $Q \subset \disk$ is a 
Euclidean quadrilateral with $ \diam(Q) \leq \epsilon \cdot \dist(Q, \z) $. 
Then the images of the vertices of $Q$ under $f$ form a quadrilateral 
with angles differing by at most $O(\epsilon)$ from the
corresponding angles of $Q$.
\end{lemma}

If we applied this directly to a general polygonal region we 
could prove that there is a quadrilateral mesh with angles 
between  $60^\circ - O(\epsilon) $ and $120^\circ + O(\epsilon)$ for 
any $\epsilon >0$, 
but we would not have  the $O(n)$ bound on the number of 
pieces. 
 Bounding the number of terms
comes from using a special decomposition of $\Omega$ and getting 
rid of the $\epsilon$'s comes from modifying the conformal map 
near the inner triangles in our mesh of $\disk$.
We will deal with the decomposition first.

A  polygonal domain $\Omega$
 is $\delta$-thick  if the corresponding prevertex set 
$\z$ is $\delta$-thick, as defined in Section \ref{mobius}. 
Equivalently, any two non-adjacent sides  of $\Omega$
have extremal distance at least $\delta$ in $\Omega$.
Extremal distance is a well know conformal invariant 
which roughly measures the distance between two 
continua compared to their diameters. For more 
details about extremal distance and thick domains, see
\cite{Bishop-time}.  

\begin{figure}[htbp]
\centerline{
 \includegraphics[height=1.25in]{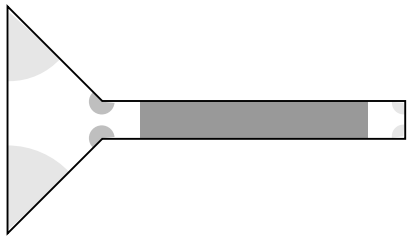}
}
\caption { \label{three-thin-types}
            A  polygon with one hyperbolic thin part (darker) and six
  parabolic thin parts.}
\end{figure}

\begin{figure}[htbp]
\centerline{
 \includegraphics[height=1.0in]{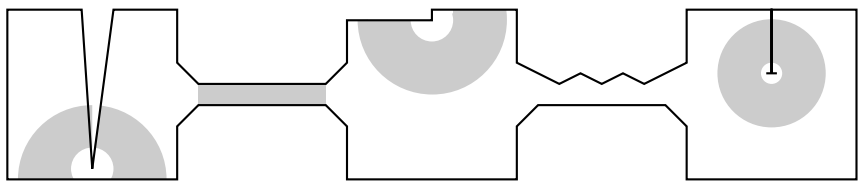}
}
\caption{\label{thick-parts}
A polygon with five hyperbolic thin parts.
This figure is not to scale.
The channel on the right is not thin because  the upper
edge is made up of numerous short edges; the extremal 
distance from any of these to the lower edge is bounded
away from zero.
}
\end{figure}

A  subdomain $\Omega' \subset \Omega$ is $\delta$-thin
if (1) $\partial \Omega' \cap \partial \Omega$ consists of two 
segments $S_1, S_2$ (each a subset of distinct edges of $\Omega$),
 (2)  $\partial \Omega' \cap 
\Omega$ consists of two polygonal arcs, each inscribed in an 
approximate circle  and (3) the extremal distance 
between $S_1$ and $S_2$ in $\Omega'$  is $\leq \delta$.
A thin part of $\Omega$ is called parabolic if the sides 
$S_1, S_2$ lie on adjacent sides of $\Omega$ is called hyperbolic 
otherwise.  See Figures \ref{three-thin-types} and \ref{thick-parts}.
 The following result is proven in \cite{Bishop-time}.

\begin{lemma}
There is an $\delta_0 >0$  and $0< C < \infty$
so that if $ \delta < \delta_0$
then the following holds.
Given a simply connected, polygonal domain $\Omega$ we can
write $\Omega$ is a union of subdomains $\{ \Omega_j\}$
belonging to two families ${\cal N }$ and $ {\cal K} $.
The elements of ${\cal N}$  are $O(\delta)$-thin polygons
and the elements of
${\cal K}$  are $\delta$-thick.
The number of edges in all the pieces put together is $O(n)$
and all the pieces can be computed in time $O(n)$ (constant
depends on $\epsilon$).
A piece can only intersect a piece of the opposite type.
Any such intersection  is a $4\delta$-thin polygon.
\end{lemma}

\begin{figure}[htbp]
\centerline{
 \includegraphics[height=1.5in]{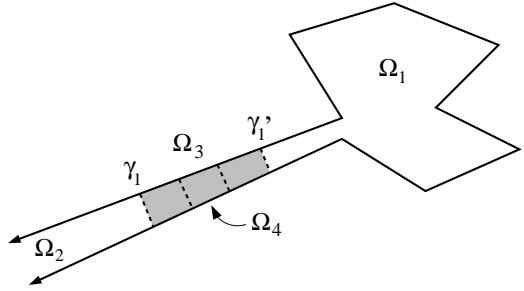}
}
\caption{\label{Overlap}
An overlapping thick piece, $\Omega_2$, and  thin piece, 
$\Omega_1$ and crosscuts $\gamma_1 = \partial \Omega_1 
\cap \Omega_2$, $ \gamma_2 = \partial \Omega_2 \cap \Omega_1$. 
The shaded region is $\Omega_3 = \Omega_1 \cap \Omega_2$.
This region 
is divided into three sections and in the center section is 
denoted $\Omega_4$.
}
\end{figure}

Suppose  $\Omega_1$ is one of the thick parts, and 
let $f : \disk 
\to \Omega_1$ be a conformal map with the origin mapping to 
a point outside of all the thin parts hitting $\Omega_1$.
 Note that $\partial 
\Omega_1 \cap \Omega$ consists of crosscuts $\{ \gamma_j\}$
 and let $\{ I_j\}$ 
be the preimages under $f$  of these  boundary arcs.
Each $\gamma_j$ has an associated crosscut $\gamma_j'$ 
that is a boundary arc of the thin part containing 
$\gamma_j$. The preimage of $\gamma_j'$ defines a crosscut
in $\disk$ whose endpoints define an interval that contains 
$A I_j$ of $I_j$ where $A \simeq \exp(\pi / 4 \delta)$.
These larger intervals are disjoint (since none of the thin 
parts intersect)  and $f(0)$ can be chosen 
so they all have length $< \pi$. 

Thus we can apply Lemma \ref{build intervals} to construct 
a domain $W_1 \subset \disk$ and a quadrilateral mesh 
on it.  Suppose  $\partial \Omega_1$ has $n_1$ sides. Since 
$\Omega_1$ is $\delta$-thick,
Lemma \ref{thick collections} implies  the mesh of $W_1$ has 
$O(n_1)$ elements,  with a constant depending on $\delta$.
Moreover, for any $\epsilon >0$ we may assume   Lemma \ref{small distortion}
applies  to all the quadrilaterals in our mesh of $W_1$ 
if we take $\epsilon = O(1/N)$  where $N$ is 
in  Lemmas \ref{pentagon mesh},  \ref{quad mesh}
 and \ref{triangle mesh}).  Thus $f(W_1) \subset \Omega_1$ 
has a mesh with $O(n_1)$ quadrilaterals, the constant depending  
on $\delta$ and $N$, which we will choose independent of $\Omega$.
If $N$ is large enough, then all angles are in the desired range,
 except possibly for the quadrilaterals corresponding to the
inner triangles. This determine the choice of $N$.

We also want to choose $\delta >0$ independent of $\Omega$.
As above, suppose $\Omega_1$ is a thick piece and that it 
intersects a thin piece $\Omega_2$.  
The intersection, $\Omega_3 = \Omega_1 
\cap \Omega_2$  is a $4 \delta$-thin 
part and can be divided into three disjoint $12 \delta$-thin  
parts as illustrated in Figure \ref{Overlap}.  Let 
$\Omega_4$ denote the ``middle'' part (the one separated
from both $\gamma_1$ and $\gamma_2$).  For points inside 
$\Omega_4$, the conformal maps of the disk to $\Omega_1$
and $\Omega_3$ are very close to each other if $\delta $ is 
small enough. The following 
result (Lemma 24 of \cite{Bishop-time}) makes this 
precise:

\begin{lemma} \label{approx map}
Suppose $f: \uhp \to \Omega_1$ is conformal. We can choose 
a conformal map $g : \uhp \to \Omega_3$ so that  for 
$z \in f^{-1}(\Omega_4)$,  and uniform $c>0, C< \infty$, 
$$|f(z)-g(z)| \leq C \exp(- c/ \delta)
\max(\diam(\gamma_1), \diam(\gamma_2)). $$
\end{lemma} 

Since $\Omega_3$ is a thin part,  we can renormalize our maps 
so that $f(i) = g(i)$ is the center of $\Omega_4$ and the 
preimages of the vertices of $\Omega_3$ 
under $g$ can be grouped into two parts: those in 
a small interval $\{ |x| <  \eta\}$ and those outside
a large interval $\{ |x| > 1/\eta\}$, where $\eta $ tends 
to zero as $\delta$ tends to zero.

The  corresponding terms of the  Schwarz-Christoffel formula 
can be  grouped as 
$$g'(w) = B  
\prod_{|z_k| < \eta} w^{\alpha_k-1}  (1- \frac {z_k}w)^{\alpha_k-1}
\prod_{|z_j| > 1/\eta}  (1- \frac w {z_j} )^{\alpha_j-1}
 \simeq B   w^{ \sum_{k: |z_k| < \eta} \alpha_k-1}  = B w^\beta ,$$
where $B$ is constant, and the dropped terms are close to 
$1$ if $\eta$ is close to $0$.  Thus $g$ approximates a 
power function.
This implies that  $g$, and hence $f$, maps the 
circular arc $\{|z|=1\} \cap \uhp$
to a smooth crosscut of $\Omega_4$ that approximates 
a circular arc  that is close to perpendicular to the boundary, 
and that $f$ followed by radial projection 
onto this arc preserves the ordering of points and  multiplies 
the distances between them by approximately a constant factor
(with error that tends to zero with $\delta$). 
Figure \ref{PowerApprox}. 
This  is one condition that determines 
our choice of $\delta$.   Another will be given in the final 
section when we mesh hyperbolic thin parts.

\begin{figure}[htbp]
\centerline{
 \includegraphics[height=1.2in]{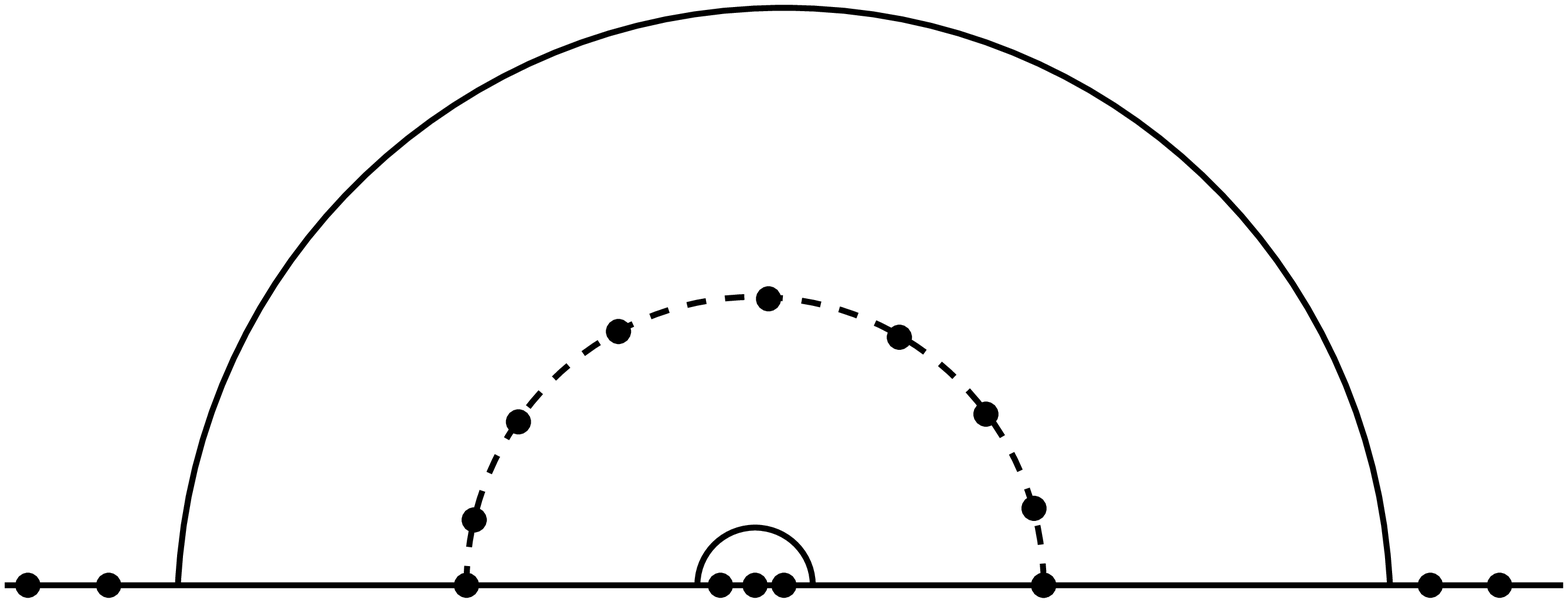}
$\hphantom{xx}$ 
\includegraphics[height=1.2in]{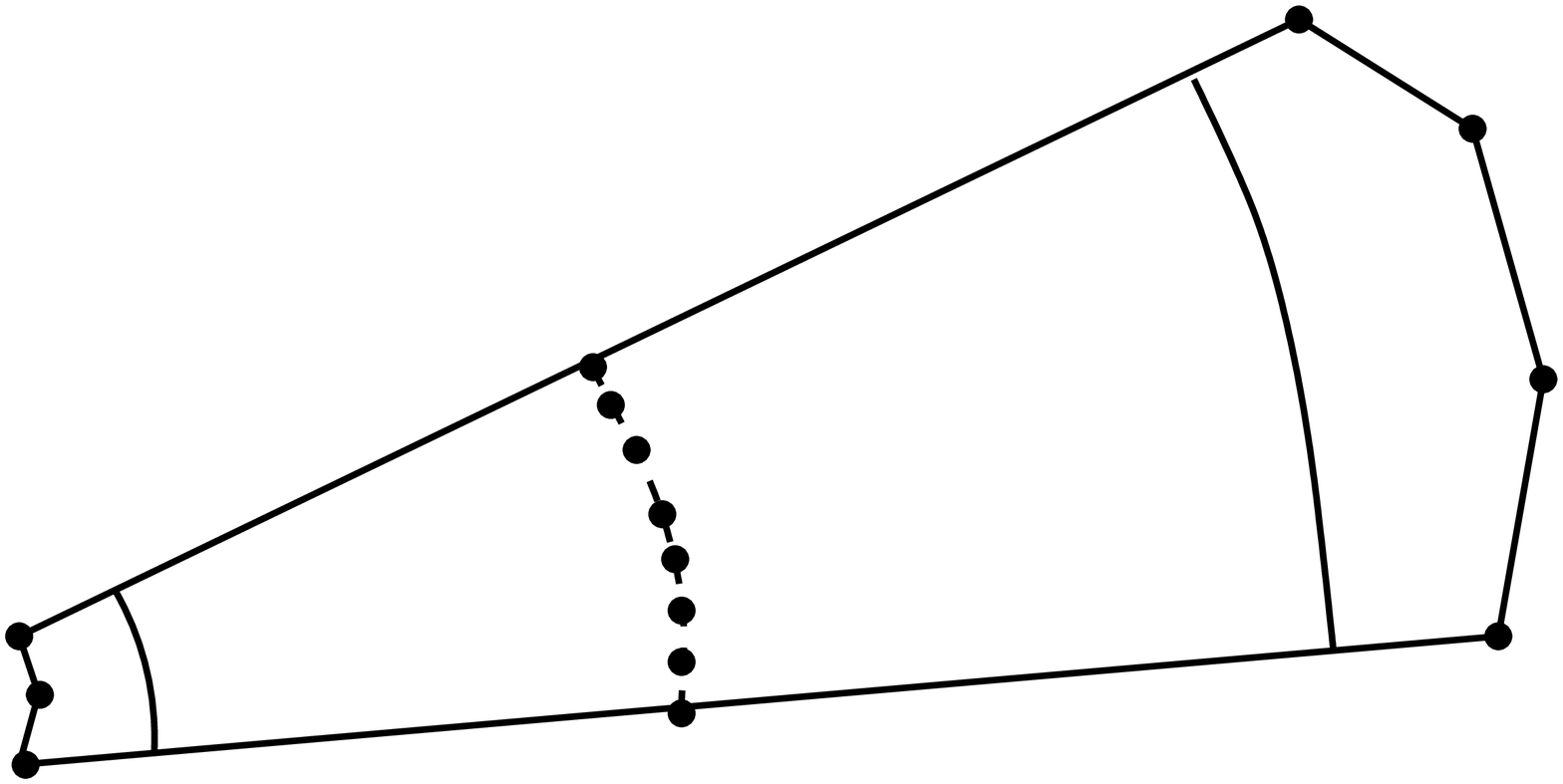}
}
\caption{ \label{PowerApprox}
Inside the middle of the overlap of a thick and thin part, 
the conformal map approximates a power function. Points 
on a circular arc in the disk are mapped to points that 
lie on an approximate circular arc and order is preserved. 
}
\end{figure}

We now transfer the mesh from $W_1$ to $f(W_1) \subset \Omega_1$.
The unmeshed portions of $\Omega$ are now all subsets of 
thin parts bounded by crosscuts that are   almost circular arcs. 
Moreover, the number of mesh vertices on each of theses 
crosscuts is the same  by (3) of Lemma \ref{build intervals}
(namely $3N+2M$ where $N$ is from Lemma \ref{pentagon mesh} and 
$M$ is from Lemmas \ref{quad mesh} and \ref{triangle mesh}).
The mesh has all angles in $[60^\circ, 120^\circ ]$,
 except those corresponding to 
 the inner triangles  in the Carleson triangles, where
they may be $O(\epsilon)$ larger or smaller.

To fix this, we replace the conformal map by a linear map 
in the inner triangles.
Map each Carleson triangle in $\disk$   used in the mesh of $W_1$, 
to the  Carleson triangle $T$  in $\uhp$ discussed in 
Section \ref{triangles}  using a M{\"o}bius transformation $\tau$.
Then $g= f_k \circ \tau^{-1}$ is a conformal map of $T$ into 
 a part of $f(W_1)$ and we transfer our mesh of $T$ outside the 
inner triangle $T_i$ via this map.  This agrees with our previous 
definition.
In the inner triangle  $T_i$ we use  the linear map
  $ h(z) = g(c) + (z-c) g'(c)$ 
to transfer the mesh. This preserves angles exactly  and so the image 
quadrilaterals have angles in $[60^\circ, 120^\circ]$. For 
quadrilaterals along the boundary of  $T_i$ 
 we apply $h$ to the vertices on $\partial T_i$ and $g$ to the vertices 
in $T \setminus T_i$.    Along the boundary 
of $T_i$, $|h(z) - g(z)| = O( \eta^2) \diam (g(T_i))$ where 
$\eta = \diam(T_i)/\dist(T_i, \z) =O(  \diam(T_i)/\diam(T))$.   Since  the
quadrilaterals meshing $T$ along $\partial T_i$ 
 have Euclidean diameter $\simeq \eta$, and 
angles all near $90^\circ$, we see that the angles of the image 
quadrilaterals also have angles near $90^\circ$ if $\eta$ is 
small enough, i.e., if the inner triangle is small enough with 
respect to the outer triangle. This determines the choice of the 
inner  triangle. 
 
This completes the proof that the desired mesh exists, 
except for meshing the thin parts, which is done in the 
next section. 
However, this is not quite a linear time 
algorithm for computing the mesh, since we have used 
evaluations of conformal maps without an estimate 
of the work involved.  We address this now.

The exact conformal map onto a general polygon
probably can't be computed 
in finite time, but we can compute an approximate 
map onto a simple $n$-gon in time $O(n)$ with a
constant depending only on the desired  accuracy.
 In \cite{Bishop-time} I show that 
a $(1+\epsilon)$-quasiconformal map from $\disk$ to 
$\Omega$ can  be computed and evaluated at $n$ points 
in time $O(n)$ where the constant depends only on 
$\epsilon$.  I will refer to \cite{Bishop-time} 
for the definition and relevant properties of 
quasiconformal mappings, but the point is that 
if $f: \disk \to \Omega$ is conformal and $g:\disk 
\to \Omega$ is the  $(1+\epsilon)$-quasiconformal approximation 
constructed in \cite{Bishop-time}, 
and if we have a Euclidean quadrilateral 
$Q$ in our mesh, then the  $g$-images of the vertices of $Q$  give angles 
that are $O(\epsilon)$ close to the angles in the $f$ image.
Thus  using $g$ to transfer the 
mesh vertices works just as well as $f$.
The fast Riemann mapping 
theorem given in \cite{Bishop-time}  implies:

\begin{thm}
Suppose we are given a  thick simply connected region $\Omega$ bounded by 
a simple $n$-gon and an $\epsilon >0$. 
We can compute the thick/thin decomposition of $\Omega$,   
the  corresponding domain $W$ and its 
 quadrilateral mesh  and 
a map $g$ on vertices of  the mesh that extends to a 
$(1+\epsilon)$-quasiconformal map of the disk to $\Omega$.
The total work is $O(n)$ where the constant may depend on 
$\epsilon$.
\end{thm}

In fact, we do not need the full strength of the result in 
\cite{Bishop-time}, giving the dependence on $\epsilon$, 
since we only need to apply the result for a small, but fixed, 
$\epsilon$. Moreover, we only need the result for thick polygons, 
which is an easier case of the theorem.

\section{Meshing the  thin parts} \label{thin}

We are now done with the proof of Theorem \ref{main}
 except for meshing thin parts.
Each such thin part is either bounded by two adjacent
edges of $\Omega$ and an almost circular crosscut $\gamma$ (the 
parabolic case) or by two non-adjacent edges and two 
almost circular crosscuts $\gamma_1,\gamma_2$ (the hyperbolic case).

We start with parabolic thin parts where the 
two adjacent edges of $\Omega$ meet at vertex $v$ with 
angle $\theta \leq 120^\circ$.  The crosscut $\gamma$ 
defines a neighborhood of $v$ in $\Omega$ that is approximately 
a sector, and we define a true circular sector $S$ with vertex 
$v$ of comparable, but smaller, size. See Figure \ref{compose-split}. 
This sector is divided into pieces using circular arcs concentric
with $v$ and radial segments,  as shown in the left of 
Figure \ref{compose-split}. There are several levels, with 
the width of the level decreasing by a factor of $2$ as we move 
away from $v$, and we split each level with radial segments 
in order to increase the number of vertices on the outer edge 
of the sector.  This can be done so that if we divide 
$S$ into four equal sectors (each of angle $\theta/4 \leq 30^\circ$)
and add extra vertices to the centers of some arcs, 
then the number of points on the outer edge in each subsector 
is the same as the number of vertices on $\gamma$ in the 
same subsector.  

If we list the points on 
$\gamma$ and on the outer edge of the sector in order, then 
corresponding points lie in the same subsector and can be 
joined by segments that make angle  between 
$90^\circ - \theta/2- \epsilon \geq 75^\circ - \epsilon$
and $105^\circ + \epsilon$ with the chords of the outer 
edge of the sector. See Figure \ref{Same-Sector}.
 A similar estimate holds for the 
chords on $\gamma$  (with a larger $\epsilon$ since $\gamma$ 
is only an approximate circle).  Here $\epsilon$ 
tends to zero  as $S$ shrinks with respect to $\gamma$. 
We simply choose a relative size for $S$ that causes these 
angles to be between $60^\circ$ and $120^\circ$.

\begin{figure}[htbp]
   \centerline{
      \includegraphics[height=1.75in]{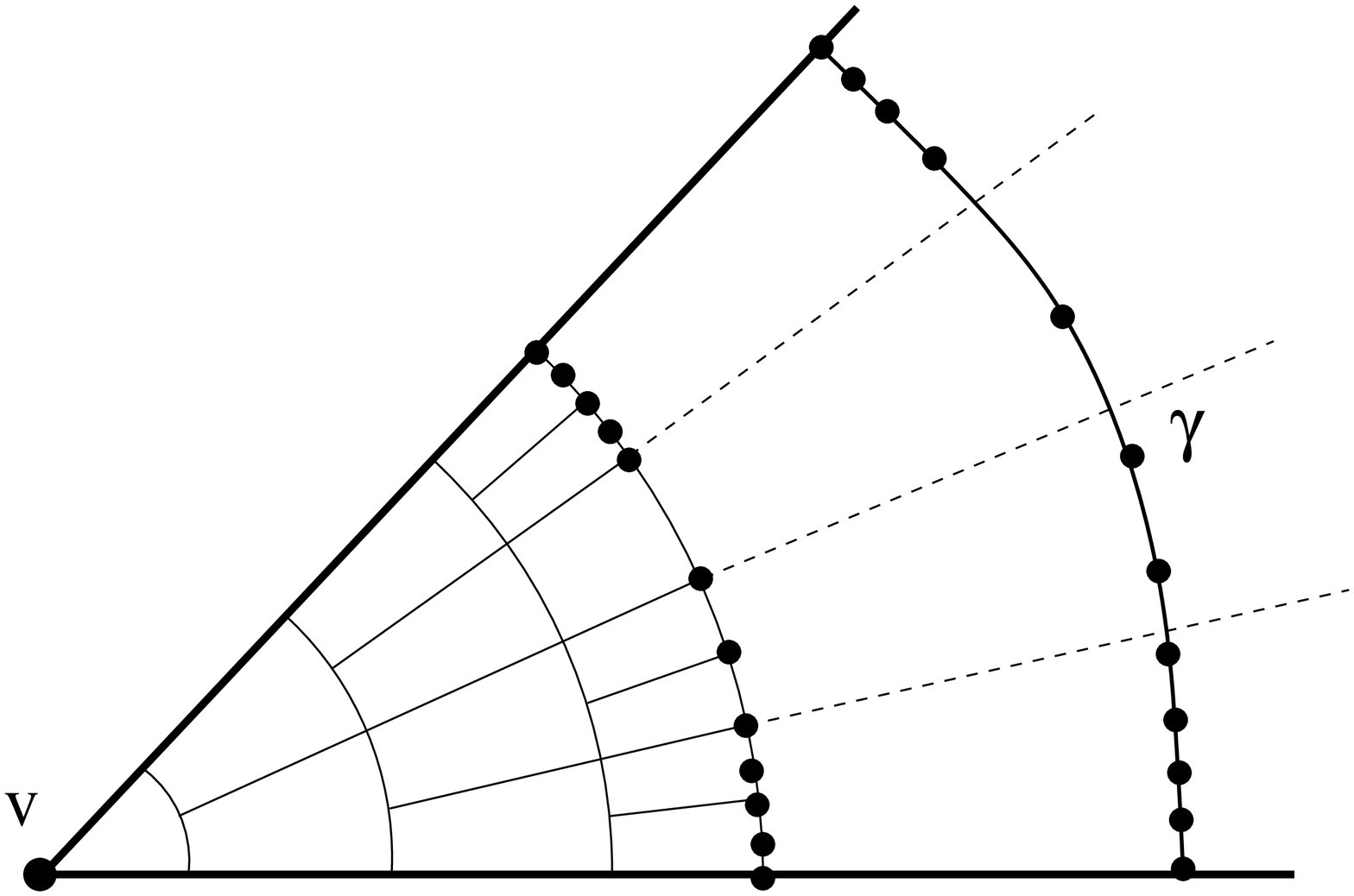}
      $\hphantom{xxx}$
      \includegraphics[height=1.75in]{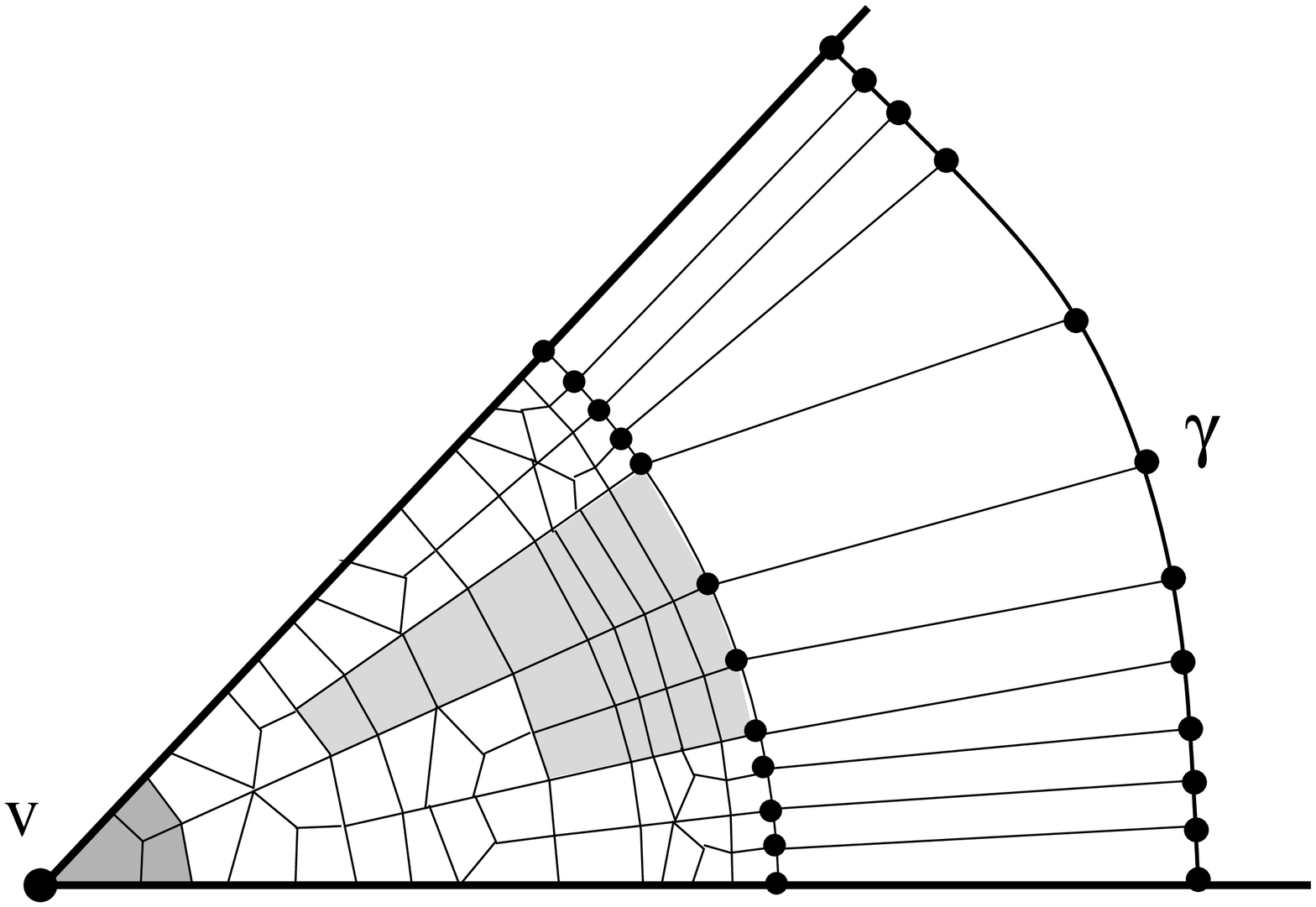}
        }
       \caption{\label{compose-split}
      The crosscut $\gamma$ is  defines a neighborhood 
      of the vertex $v$. We define a sector of comparable 
      size and partition the sector, so that the number of 
      vertices on the outer edge approximates the number of 
      points on the crosscut $\gamma$.  The pieces are then 
      meshed: Mesh 1 is used in dark shaded region, Mesh 2 (or 
      it reflection) the white regions and segments only in 
      the lighter shaded regions. The number of 
      vertices on the outer edge is exactly the number on $\gamma$
      and corresponding points are joined by segments. 
   }
\end{figure}

\begin{minipage}{3in}                             
   \centerline{ 
	\includegraphics[height=1.75in]{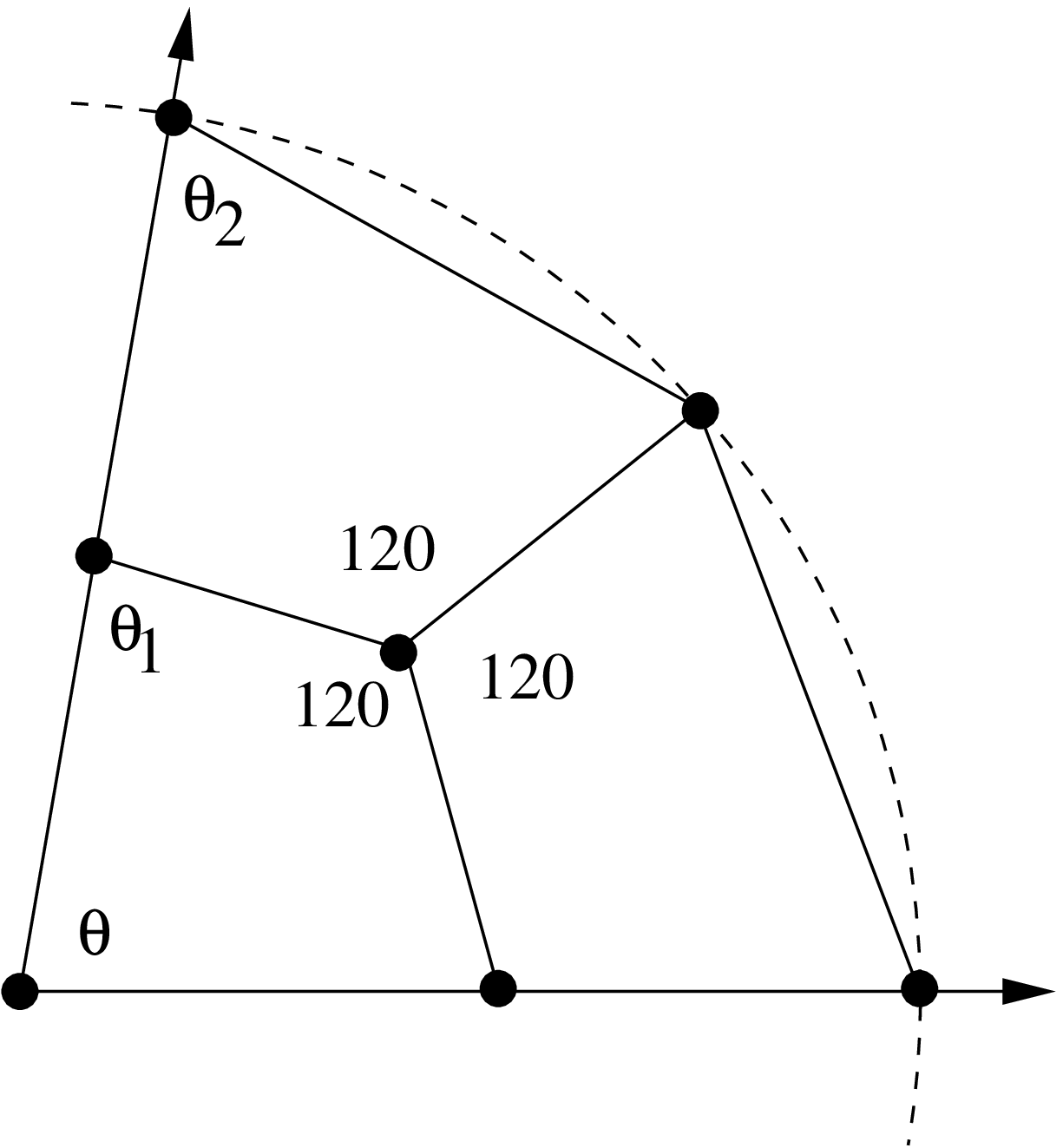}
	}
  \centerline{Mesh 1} 
\begin{eqnarray*}
 && 0 < \theta \leq 120 \\
 &&  60 \leq \theta_1 = 180 - 60 - \frac 12 \theta \leq 120 \\
 && 60 \leq \theta_2 = \frac 12(180-\frac 12 \theta)  \leq 90 
\end{eqnarray*}
\end{minipage} 
\begin{minipage}{3in} 
   \centerline{ 
	\includegraphics[height=1.5in]{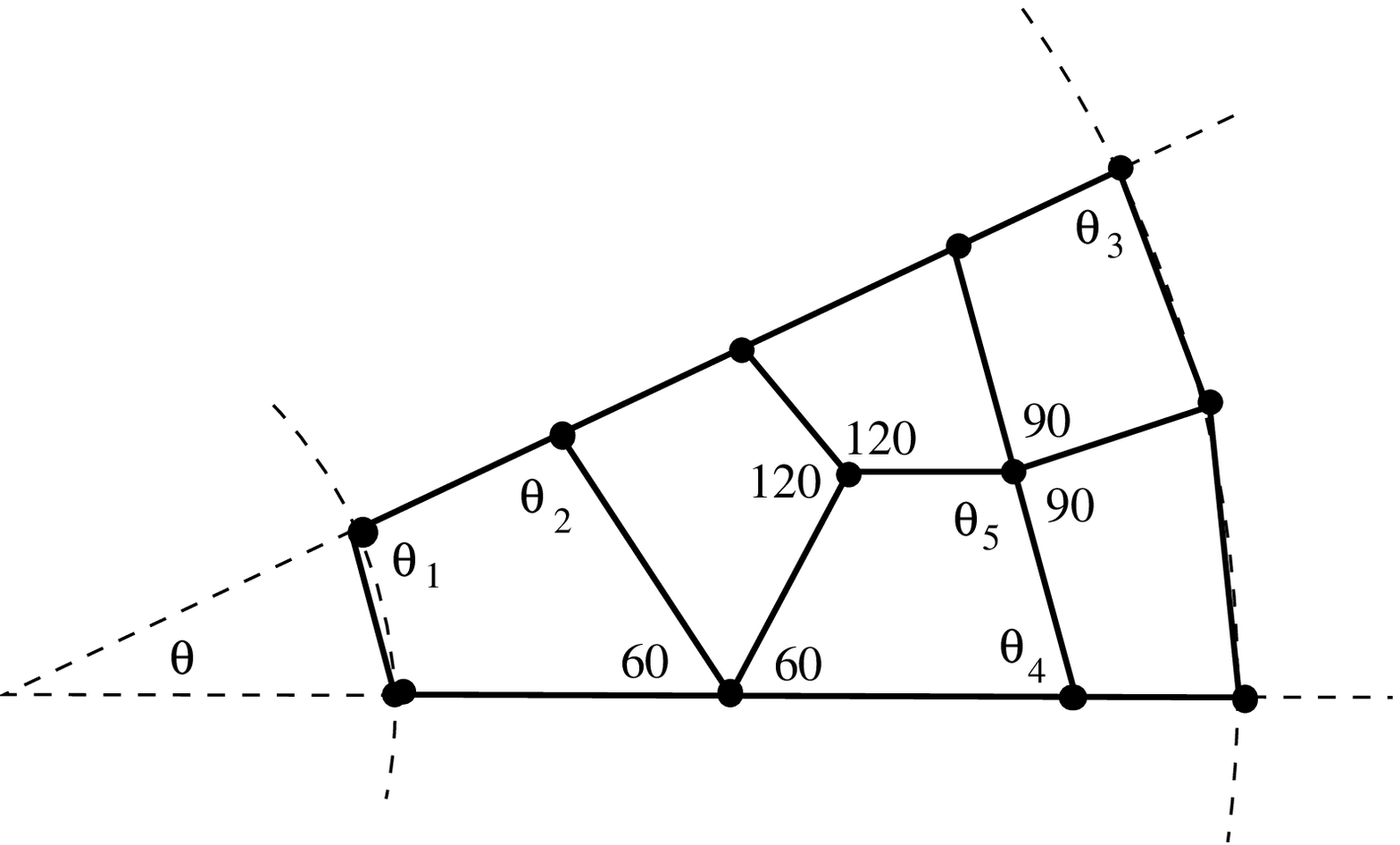}
	}
   \centerline{Mesh 2} 
{\small{
\begin{eqnarray*}
 && 0 \leq   \theta \leq 60 \\
 && 90 \leq \theta_1= 180-\frac 12(180-\theta) \leq 120 \\
 &&  60 \leq   \theta_2 = 360 - 60 - 2 \theta_1 
                    \leq 120 \\
 && 75   \leq \theta_3 
               = 90 - \frac 14 \theta \leq 90  \\
  && 60  \leq \theta_4 =  90 - \frac 12 \theta \leq 90 \\
  && 90  \leq \theta_5 = 360-120-60-\theta_4 \leq 120\\
\end{eqnarray*}
}}
\end{minipage} 

\begin{figure}[htbp]
   \centerline{
      \includegraphics[height=1.75in]{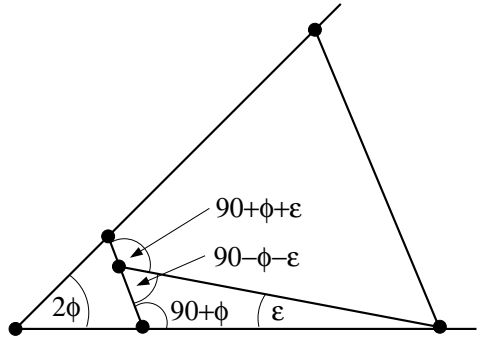}
        }
       \caption{\label{Same-Sector}
        The connecting segments between $\gamma$ and the outer edge
        of $S$ lie inside a sector of angle $2 \phi \leq 30$. If the 
        $S$ is small enough compared to $\gamma$ the angle marked 
        $\epsilon$ is as small as we wish, say $\epsilon < 10^\circ$.
        Then the angles formed with the chords of the outer edge of $S$ 
        are between $65^\circ$ and $115^\circ$. The angles with the 
        chords along $\gamma$ are slightly smaller/larger since $\gamma$ is 
        only an approximate circle, but the difference is as small as 
        wish by taking the  parameter $\delta$ in our thick/thin 
        decomposition small enough.
   }
\end{figure}

We then have to mesh $S$ so that the mesh vertices on the outer edge 
are exactly the ones given above.  We  do this by applying the 
illustrated  constructions in each part of the sector.
Mesh 1 is  used only in the piece adjacent to $v$ and the equations 
below the figure show that all the new angles are in the correct
range. Mesh 2 (or its reflection) are used in all the pieces 
that have one more vertex on their outer edge than on the inner
edge (we use reflections to make the vertices on the radial 
edges match up). Otherwise we simply use  chords of circles
concentric with $v$ to connect edge vertices  of parts where
mesh 2 was used. See the right side of Figure \ref{compose-split}.

If the interior angle at $v$ is $120^\circ \leq \theta < 240^\circ$
then we bisect the angle as part of our partition of the sector.
If $240^\circ \leq \theta \leq 360^\circ$, then we trisect the 
angle. See Figure \ref{big angle}.

\begin{figure}[htbp]
   \centerline{
      \includegraphics[height=1.25in]{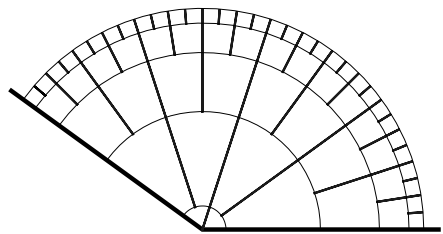}
        }
       \caption{\label{big angle}
    If the vertex has interior angle between $120^\circ$ and 
    $240^\circ$ then we bisect the angle as part of the sector 
    partition and mesh each piece as before. 
   }
\end{figure}

 A hyperbolic thin part  is bounded by two straight 
line segments in $\partial \Omega$  and two almost 
circular crosscuts $\gamma_1, \gamma_2$. Both crosscuts
contain the same number, $P$, of vertices from the meshes 
of the corresponding thick pieces. 
If the two straight sides are parallel or lie on lines that 
intersect with small angle, then just connecting each 
point on $\gamma_1$ to the corresponding point on $\gamma_2$
will give angles in the desired range. In general, however, 
this is not the case, but is easily fixed by adding a
bounded number of circular crosscuts separating 
$\gamma_1, \gamma_2$ and using a polygonal chain with 
vertices on these crosscuts to connect each vertex on
$\gamma_1$ to the corresponding vertex on $\gamma_2$. It
is easy to see that this can be done with angles close
to $90^\circ$ if the number of intermediate crosscuts is 
large enough and $\delta$ (the degree if thinness) is small enough.
See Figure \ref{HyperConnect}. 
This places an additional constraint on the choice of $\delta$. 

\begin{figure}[htbp]
   \centerline{
      \includegraphics[height=1.25in]{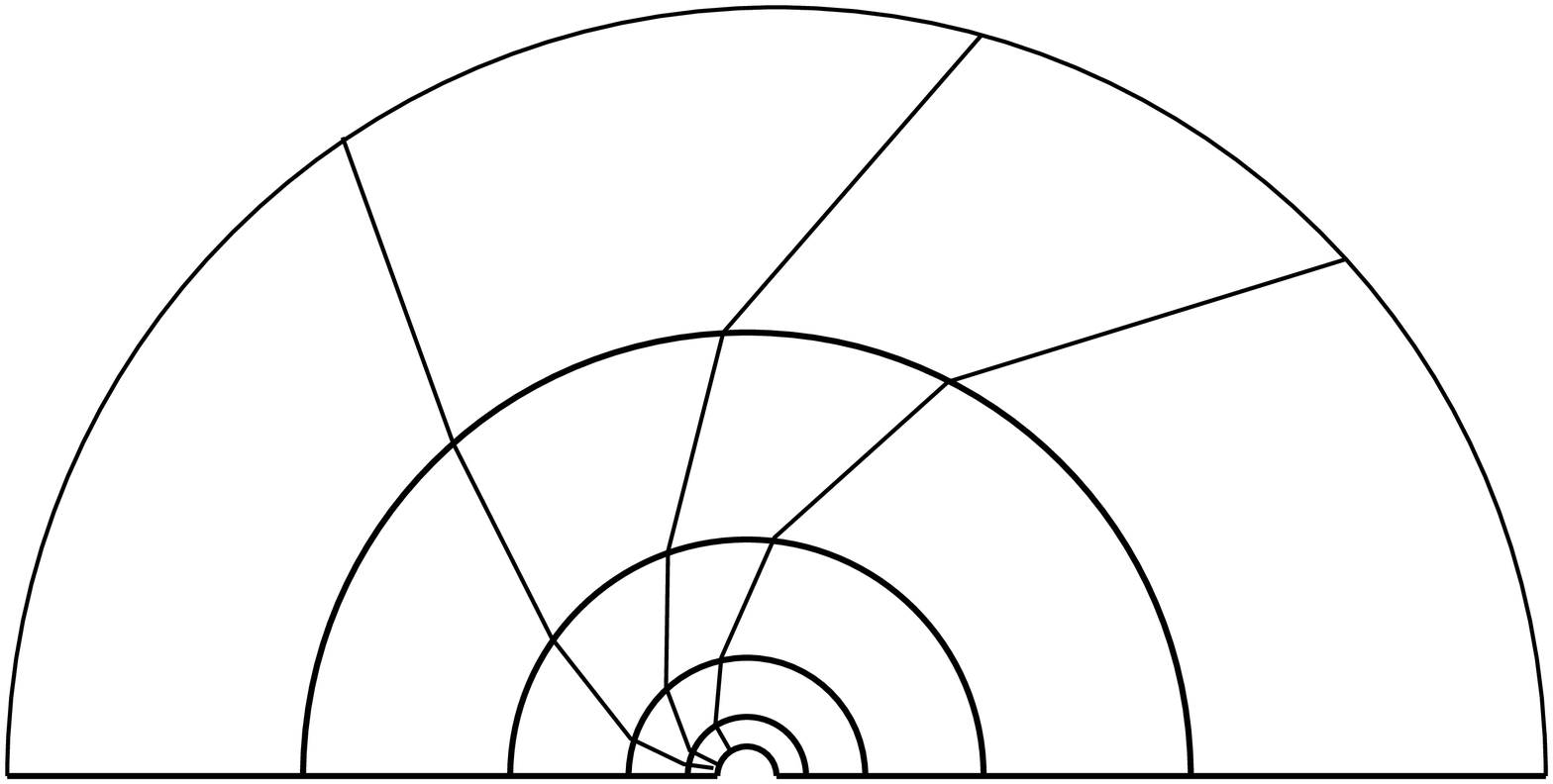}
        }
 \vskip .3in
   \centerline{
      \includegraphics[height=.5 in]{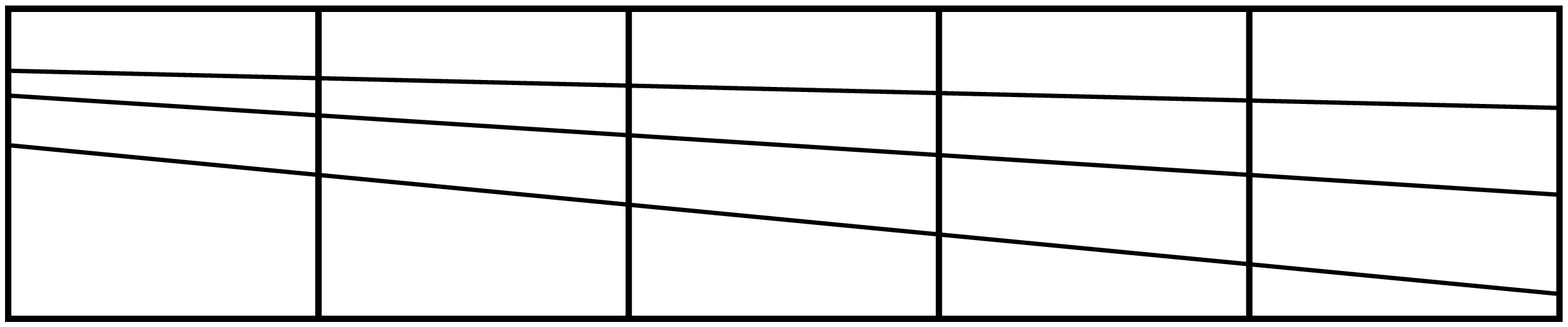}
        }
       \caption{\label{HyperConnect}
      By adding a bounded number of circular crosscuts 
      to a hyperbolic thin part we can connect any $P$ 
      points on $\gamma_1$ to any $P$ points on $\gamma_2$ 
      with a mesh using angles near $90^\circ$. 
      The arcs look like logarithmic spirals. Indeed, we can 
      think of this mesh as approximating the image of the 
      lower  picture under the complex  exponential map. 
   }

\end{figure}

In addition to the angle bounds, every quadrilateral in the
construction can be chosen to  have bounded geometry (i.e.,
all four edges of comparable length with uniform constants)
except in two cases. First, when we mesh a parabolic thin part with 
angle $\theta \ll 1$, the piece containing the vertex has 
two sides with length only $O(\theta)$ as long as the other two.  
Second, when meshing a hyperbolic thin part we use 
long, narrow pieces, but  if the long
sides have extremal distance $\delta$,   we can refine the 
mesh by subdividing each  such piece into  $O(1/\delta)$
bounded geometry quadrilaterals.
Thus if the hyperbolic thin parts of $\Omega$ have ``thinnesses''
$\{\delta_k\}$, then we can mesh $\Omega$ by $O(n+\sum_k \delta_k^{-1})$
quadrilaterals with angles in $[60^\circ, 120^\circ]$ and 
bounded geometry, except for the pieces containing vertices with
 small angles.  If $\Omega$ has no small angles, then this gives 
the smallest (up to a constant factor), bounded geometry mesh 
of $\Omega$.

\bibliography{optimal} 
\bibliographystyle{plain}
\end{document}